\documentclass[11pt,a4paper]{article}

\usepackage[T1]{fontenc}
\usepackage{newtxtext}
\usepackage[margin=1in]{geometry}
\usepackage{amsmath,amsthm,mathtools}
\usepackage{newtxmath}
\usepackage{graphicx,booktabs,array,tabularx,threeparttable}
\usepackage{caption,microtype,xcolor}
\usepackage{placeins}
\usepackage[round,authoryear]{natbib}
\usepackage[colorlinks=true,linkcolor=blue!45!black,
  citecolor=blue!45!black,urlcolor=blue!45!black]{hyperref}

\graphicspath{{./}}
\newtheorem{theorem}{Theorem}[section]
\newtheorem{proposition}[theorem]{Proposition}
\newtheorem{corollary}[theorem]{Corollary}
\newtheorem{lemma}[theorem]{Lemma}
\newtheorem{assumption}[theorem]{Assumption}
\theoremstyle{definition}

\theoremstyle{remark}
\newtheorem{remark}[theorem]{Remark}

\newcommand{\Q}{\mathbb{Q}}
\newcommand{\Pp}{\mathbb{P}}
\newcommand{\R}{\mathbb{R}}
\newcommand{\E}{\mathbb{E}}
\newcommand{\one}{\mathbf{1}}
\newcommand{\cI}{\mathcal{I}}
\newcommand{\cM}{\mathcal{M}}
\newcommand{\cl}{\operatorname{cl}}

\hypersetup{
  pdftitle={The Physical Crash Frontier: What Finite Option Quotes Can and Cannot Reveal},
  pdfauthor={Jirong Zhuang},
  pdfsubject={Sharp identification of option-implied physical crash risk},
  pdfkeywords={physical crash frontier, physical crash probability,
    partial identification, bid and ask quotes, pricing kernel}
}

\title{\textbf{The Physical Crash Frontier:\\ What Finite Option Quotes Can and Cannot Reveal}}
\author{Jirong Zhuang\\
\normalsize Department of General Education, Guangzhou Polytechnic University\\
\normalsize 1342 Shiliang Road, Panyu District, Guangzhou, Guangdong 511483, China\\
\normalsize\texttt{zhuangjirong6552@outlook.com}}
\date{October 2026}

\begin{document}
\maketitle

\begin{abstract}
Physical crash probabilities recovered from option prices depend on a pricing kernel and on a risk-neutral distribution that finitely many bid and ask quotes do not identify.  For a power utility investor, we characterize the pairs of physical crash probability and expected loss below the crash threshold that the quotes admit; the boundary of this set is the physical crash frontier.  Both coordinates are ratios of moments, yet when the index is bounded above the closed set is convex, and under an interior regularity condition second-order cone programs compute it exactly at the benchmark risk aversion of two.  In a decade of weekly S\&P 500 cross sections, the quotes beyond the two puts nearest a 10 percent decline shrink the range of admissible crash probabilities by about 80 percent, yet its upper end remains two to three times its lower end.  Under strict quote slack and the interior-mass condition stated below, removing the cap without further tail control drives the lower bounds on crash probability and unconditional shortfall to zero for risk aversion above one.  A vanishing probability far in the right tail inflates the normalizing moment while every quote remains inside its spread.  In this setting, a positive floor requires additional tail information, supplied here by the support cap.
\end{abstract}



\section{Introduction}
\label{sec:introduction}

Severe market declines are hard to measure for the same reason they matter: they are rare.  Over a horizon of a few weeks, realized returns say almost nothing about the probability of a crash or about how deep the losses would be.  Option markets price this risk every day, and puts struck below any given threshold trade at observable bid and ask quotes, so it is natural to read crash probabilities from them.  The difficulty is that option prices are state prices.  They measure what investors pay for protection, and the same fear that makes the measurement urgent also raises that price, so the probabilities implied by option prices overstate physical crash risk, most severely in times of stress \citep{martinshi2026forecasting}.  Getting from prices to physical probabilities requires a restriction on some investor's preferences.  This paper asks how much the contracts that actually trade reveal about physical crash risk once that restriction is in place.

The usual approach fits a smooth curve through the observed option prices and differentiates it twice to obtain a single risk-neutral density.  The fitted curve blends what the market says with assumptions about strikes that never traded, and it gives no indication of how much of the answer comes from each.  We take a different route and keep the ambiguity in view.  For a given cross section, we collect every risk-neutral distribution that prices each quote inside its spread, has the forward as its mean, and is supported on a bounded interval.  \citet{zhuang2026likelydeepsharpjoint} characterizes the sets of crash probability and crash depth that these distributions generate under the risk-neutral measure and calls their boundary the risk-neutral crash frontier.  Those sets describe the price of crash insurance, and a price of insurance mixes the probability of the loss with the compensation for bearing it.

This paper takes the remaining step.  Each admissible distribution is passed through the pricing kernel of a power utility investor with relative risk aversion \(\gamma\) who holds the market.  \citet{martin2017expected} introduced this restriction for the log investor, and \citet{martinshi2026forecasting} generalized it to arbitrary risk aversion and showed that the resulting crash probability bounds forecast realized crashes.  What comes out is the physical distribution as that investor sees it.  Every statement in the paper is conditional on that kernel: the results describe what the quotes imply for her.

For each admissible distribution, two numbers summarize the physical crash risk it implies: the probability that the index ends below the threshold, and the expected loss below the threshold, which we call the shortfall.  Collecting these pairs across all admissible distributions gives an identified set in the sense of \citet{manski2003partial}, and we call its boundary the physical crash frontier of the cross section.  It separates the crash scenarios the quotes admit from those they rule out.  The closed set is compact and convex.  At the benchmark exponent two, an interior regularity condition that can be checked cross section by cross section gives an exact finite second-order cone representation: every point it admits is generated or approached arbitrarily closely by admissible distributions.  Given the kernel and the support condition, the frontier is therefore exactly what the quotes pin down; anything tighter is an assumption.  Setting risk aversion to zero recovers, up to the support bound, the risk-neutral sets of \citet{zhuang2026likelydeepsharpjoint}, so the construction places them in a family indexed by risk aversion.

The joint set matters because risk scenarios come in pairs.  A stress test that specifies a crash probability together with an expected loss needs to know whether admissible distributions can jointly attain or approach the pair.  The two marginal intervals cannot say, because their endpoints may come from different distributions.  The frontier excludes combinations outside the joint set, cutting most deeply where the shortfall is stressed, and checking any candidate pair amounts to one conic feasibility problem (Corollary~\ref{cor:scenario-audit}).

Moving from prices to physical probabilities changes the mathematics, and the change turns out to matter economically.  Under the risk-neutral measure, crash probability and expected loss are linear in the distribution, and their identified sets come from linear programs.  Under the physical measure they are ratios.  Our benchmark sets risk aversion equal to two, so the change from the risk-neutral to the physical measure weights each state by the square of the normalized market payoff and the physical crash probability becomes the second moment of the crash states divided by the second moment of all states.  The denominator is the problem.  It involves every state, including payoffs far above the highest quoted strike, where nothing trades, so two distributions that price every quote identically can disagree about it, and under the perturbability conditions of Proposition~\ref{prop:unbounded-tail} the denominator has no finite upper bound: a vanishing amount of probability placed at an ever larger payoff pushes it toward infinity and the physical crash probability toward zero while every quote stays inside its spread.  We call this the dilution channel.

This channel shows why, under those perturbability conditions, finite vanilla quotes alone cannot keep physical crash probability away from zero when $\gamma>1$.  Additional control of the far right tail is needed to prevent this mechanism; the support cap maintained here supplies such control.  Section~\ref{sec:tail-theory} states its role precisely.  Conditional depth is the ratio of shortfall to probability, so the global normalizing moment cancels.  Its feasible range can nevertheless change with the tail restriction through the quote and forward constraints.

The empirical part of the paper asks how much the quotes actually narrow these sets.  They narrow them a great deal, and most of the narrowing comes from a handful of strikes.  We use more than a thousand S\&P 500 cross sections, sampled weekly from 2013 through 2023 at three fixed horizons close to one month.  Expanding the quote set from the two puts nearest a 10 percent decline to all retained out-of-the-money quotes shrinks the interval of admissible crash probabilities by a median of about 80 percent.  Most of that shrinkage arrives with the six additional puts closest to the threshold, and even with every quote in place the upper end of the interval remains two to three times the lower end.  The intervals rise in every major stress episode of the sample, from the August 2015 selloff through early 2018 and the COVID-19 crisis to the 2022 selloff.  In March 2020 both ends jump at once: stress raises the level of crash risk and at the same time widens the range of crash probabilities the quotes can accommodate.

Section~\ref{sec:literature} places the analysis in the literature.  Section~\ref{sec:framework} develops the identification framework and the exact computation, Section~\ref{sec:empirical} describes the data and design, Sections~\ref{sec:marginal-results} and~\ref{sec:joint-results} report the marginal and joint results, and Section~\ref{sec:conclusion} concludes.  Proofs, the log-investor benchmark, and supplementary evidence are in the appendices.

\section{Related Literature}
\label{sec:literature}

With a continuum of option prices, the risk-neutral distribution is identified under regularity conditions \citep{breedenlitzenberger1978}.  Empirical work builds on that identification.  \citet{pan2002jump} estimates jump risk premia from the joint dynamics of spot and option prices in a parametric model, \citet{backuschernovmartin2011disasters} compare disaster distributions implied by index options with those calibrated to macroeconomic data, and \citet{bollerslevtodorov2011tails} separate compensation for tail events from the events themselves, reading risk-neutral jump tails from short-maturity options and physical jump tails from high-frequency returns.  The premium component of these tail measures forecasts returns \citep{bollerslevtodorovxu2015}.  Each approach reaches its point estimates by adding structure, whether a parametric model or a tail extrapolation.  We start one step earlier and ask what is identified before that structure is imposed.

Physical crash probabilities can be approached from three directions.  The most direct is to ask: survey respondents report subjective crash probabilities far above historical crash frequencies \citep{goetzmannkimshiller2016}.  Recovery arguments attempt to extract physical probabilities from state prices while replacing preference restrictions with structural ones: \citet{hansenscheinkman2009} develop the operator decomposition that underlies these arguments, \citet{ross2015recovery} establishes recovery for a Markov state with a transition-independent kernel, and \citet{borovickahansenscheinkman2016} show that the recovered measure differs from the physical one whenever the kernel has a nondegenerate martingale component, so an explicit restriction is unavoidable.  A third route supplies the restriction directly.  \citet{aitsahalialo2000} and \citet{jackwerth2000riskaversion} infer risk aversion by comparing option-implied state prices with return histories, \citet{rosenbergengle2002} estimate time-varying pricing kernels projected on market returns, and \citet{blisspanigirtzoglou2004} calibrate utility so that the transformed densities best forecast realized returns.

The restriction maintained here comes from \citet{martin2017expected}, whose benchmark investor has log utility and holds the market.  \citet{martinshi2026forecasting} generalize his expectation identity to arbitrary risk aversion and convert risk-neutral crash probabilities into physical ones.  They document that the risk-neutral versions overstate crash risk most severely in crises, and they calibrate risk aversion to two on market returns.  Their stock-level bounds rely on Fr\'echet--Hoeffding arguments because marginal option prices do not identify the joint distribution of a stock with the market.  For the market index itself the dependence problem disappears, and identification uncertainty comes instead from finite strikes, bid--ask intervals, and the unobserved right tail.

A finite cross section turns identification into a generalized moment problem, and the idea of answering with a set is old.  \citet{lo1987semiparametric} derives upper bounds for option prices and expected payoffs from moment restrictions, and \citet{bassett1997bounds} bounds the probabilities of future prices by the set of distributions that could have generated the observed option values.  \citet{bertsimaspopescu2002} derive best-possible bounds on option prices from moment information by semidefinite programming, and \citet{kingkoivupennanen2005} compute arbitrage-free price bounds for a claim from the prices of traded options through convex duality.  \citet{davishobson2007} characterize when a finite set of option prices is consistent with an arbitrage-free model.  \citet{coxobloj2011} derive model-free bounds for barrier claims from call and digital call quotes, where mass exactly at a barrier changes the answer.  In a single-maturity bid--ask martingale transport problem, \citet{liangetal2026} derive a square-root convergence bound for upper-semicontinuous payoffs of bounded variation when the mid-marginal has a bounded density.

We read the resulting solution sets as identified sets in the sense of \citet{manski2003partial} and describe them through their support functions, following \citet{beresteanumolinari2008}.  The closest computational relatives are \citet{barratttuckboyd2024}, who bound the cumulative distribution and distributional risk measures over risk-neutral probabilities on a discretized price grid, and \citet{vaneekelen2023}, who extends generalized moment duality to conditional expectations, the fractional form that conditional depth takes here.  The direct predecessor is \citet{zhuang2026likelydeepsharpjoint}, which introduces the risk-neutral crash frontier, the boundary of the sharp joint set of risk-neutral crash probability and conditional depth implied by finite bid--ask quotes, and shows that separate marginal extrema need not be jointly attainable.  We differ in the object that is identified: the bounds here concern physical crash risk under the benchmark kernel.  Positive risk aversion makes the transform to physical probabilities normalize by a moment of the unknown distribution.  At $\gamma=2$, the normalization turns linear objectives into ratios and leads to second-order cone programs.  Controlling the normalizing moment requires additional tail information, supplied here by a support cap.  That condition is the source of the tail identification problem studied below.

\section{From Option Quotes to Sets of Physical Crash Probabilities}
\label{sec:framework}

We work one cross section at a time.  This section defines the admissible risk-neutral distributions, the physical crash quantities they imply, and the identified sets, and then shows how to compute the sets exactly.

\subsection{Normalized prices and admissible measures}
\label{sec:normalization}

Fix a valuation date \(t\) and an option maturity \(T\).  Let \(D\) denote the discount factor and \(F\) the forward price of the cross section.  The random variable \(S_T\) is the terminal index level, and we work with the forward-normalized terminal level
\begin{equation}
  X=\frac{S_T}{F}.
  \label{eq:normalized-terminal}
\end{equation}
The normalization makes \(X\) dimensionless, and under the normalized risk-neutral measure \(\Q\) its first moment equals one.

Quote \(i\) has strike \(\kappa_i\) in index points and normalized strike \(k_i=\kappa_i/F\).  We write \(\underline c_i\) and \(\overline c_i\) for the normalized bid and ask of a call payoff at that strike, so a retained call quote imposes
\begin{equation}
  \underline c_i
  \leq \E_{\Q}\!\left[(X-k_i)^+\right]
  \leq \overline c_i.
  \label{eq:quote-interval}
\end{equation}
Put quotes enter the same form through put--call parity: their normalized endpoints are the put bid and ask divided by \(DF\), with \(1-k_i\) added to each endpoint.

Let \(\cI\) denote a finite set of retained quote indices, and let \(B>1\) be an upper support bound.  The admissible set is
\begin{equation}
\begin{split}
  \cM(\cI)
  =\biggl\{\Q:\;&
  \Q\ \text{is a probability measure on }[0,B],\\
  &\E_{\Q}[X]=1,\quad
  \underline c_i\leq
  \E_{\Q}\!\left[(X-k_i)^+\right]
  \leq\overline c_i\quad\text{for every }i\in\cI
  \biggr\}.
  \label{eq:admissible-measures}
\end{split}
\end{equation}
In general \(\Q\) is the normalized state-price measure for maturity \(T\), and when discounting is deterministic it is the usual risk-neutral measure.  Support on \([0,B]\) is a maintained economic restriction whose role Section~\ref{sec:tail-theory} isolates.  The set \(\cM(\cI)\) also depends on \(B\) and on \((D,F)\), which are held fixed throughout, so we suppress them and keep only the quote set \(\cI\), the object the empirical design varies.

\subsection{The power utility investor and the crash quantities}
\label{sec:physical-transform}

\begin{assumption}[Power utility pricing kernel]
\label{ass:kernel}
The investor has constant relative risk aversion \(\gamma>0\), is fully invested in market wealth, and is marginal in the market.  Her gross return on market wealth is \(R_m=cX\) for a constant \(c>0\), gross returns are strictly positive, and her stochastic discount factor is proportional to \(R_m^{-\gamma}\).
\end{assumption}

For each \(\gamma>0\) and admissible \(\Q\), let \(M_\gamma(\Q)=\E_{\Q}[X^\gamma]\).  Under Assumption~\ref{ass:kernel}, the physical measure \(\Pp_\gamma^{\Q}\) perceived by the investor satisfies
\begin{equation}
  \frac{d\Pp_\gamma^{\Q}}{d\Q}
  =\frac{X^\gamma}{M_\gamma(\Q)}.
  \label{eq:physical-transform}
\end{equation}
Normalized state pricing makes \(d\Q/d\Pp\) proportional to the stochastic discount factor, hence to \(X^{-\gamma}\), and inverting and normalizing under \(\Q\) gives the displayed density.  The transform restates, for the normalized terminal level, the expectation identity that \citet{martin2017expected} derives for a fully invested log investor and \citet{martinshi2026forecasting} extend to arbitrary risk aversion.  The support and forward restrictions imply \(0<M_\gamma(\Q)\leq B^\gamma\), so the transform is well defined for every admissible measure.\footnote{In Assumption~\ref{ass:kernel}, the relation \(R_m=cX\) treats \(X\) as market wealth after deterministic dividend and numeraire adjustments.  Under Assumption~\ref{ass:regularity}, admissible measures with mass at zero enter as a closure relaxation of the positive-return model and leave the closed identified sets unchanged (Lemma~\ref{lem:positive-return-closure}).}

Fix a threshold \(K\in(0,1)\), held constant throughout, and define the crash event \(\{X\leq K\}\).  For each exponent \(r\geq0\), define the truncated moment
\[
A_r(\Q)=\E_{\Q}\!\left[X^r\one\{X\leq K\}\right].
\]
The crash probability and the unconditional shortfall, the expected loss below the threshold, are
\begin{align}
  q_{\gamma,K}(\Q)
  &=\Pp_\gamma^{\Q}(X\leq K)
    =\frac{A_\gamma(\Q)}{M_\gamma(\Q)},
  \label{eq:general-q}\\
  \ell_{\gamma,K}(\Q)
  &=\E_{\Pp_\gamma^{\Q}}\!\left[(K-X)^+\right]
    =\frac{K A_\gamma(\Q)-A_{\gamma+1}(\Q)}
           {M_\gamma(\Q)}.
  \label{eq:general-ell}
\end{align}
When \(A_\gamma(\Q)>0\), conditional depth is the ratio
\begin{equation}
  s_{\gamma,K}(\Q)
  =\E_{\Pp_\gamma^{\Q}}\!\left[K-X\mid X\leq K\right]
  =K-\frac{A_{\gamma+1}(\Q)}{A_\gamma(\Q)}
  =\frac{\ell_{\gamma,K}(\Q)}{q_{\gamma,K}(\Q)}.
  \label{eq:general-depth}
\end{equation}
Thus \(\ell_{\gamma,K}\) measures the unconditional distance below \(K\), while \(1-K+s_{\gamma,K}\) is the conditional decline from the normalized forward.\footnote{Throughout, shortfall refers to this expected loss below the fixed threshold, not to the quantile-based risk measure that shares the name.}  Setting \(\gamma=0\) in \eqref{eq:general-q}--\eqref{eq:general-ell} makes the transform the identity and recovers the risk-neutral crash probability and normalized put value, the \(p_K\) and \(L_K\) whose joint sets \citet{zhuang2026likelydeepsharpjoint} characterizes.  That paper places no upper bound on the state, so the \(\gamma=0\) sets here are the support-capped counterparts of its joint sets.  Positive \(\gamma\) reweights each admissible measure by \(X^\gamma\), undoing the premium that state prices attach to low-payoff states.  Probability mass shifts toward high-payoff states, so physical crash probabilities do not exceed their risk-neutral counterparts.  The price of this correction is that the reweighting depends on \(M_\gamma(\Q)\), a global moment of the unknown distribution.

We record mass exactly at the threshold in a separate coordinate.  Moving mass across \(K\) changes every retained vanilla payoff continuously but changes the indicator of the closed event \(\{X\leq K\}\) discontinuously.  Let \(a_K=\Q\{K\}\) and let \(A_r^{<}(\Q)=\E_{\Q}[X^r\one\{X<K\}]\) denote the truncated moment over the open event.  Then
\begin{equation}
  A_r=A_r^{<}+K^r a_K,
  \qquad
  K A_\gamma-A_{\gamma+1}
  =K A_\gamma^{<}-A_{\gamma+1}^{<},
  \label{eq:general-threshold-atom}
\end{equation}
so an atom at the threshold adds \(K^\gamma a_K/M_\gamma(\Q)\) to the crash probability and nothing to the shortfall.

\subsection{The identified sets and their geometry}
\label{sec:identified-set}

Each admissible measure delivers one probability--shortfall pair.  The identified set collects every pair that admissible measures attain or approach:
\begin{equation}
  \Theta^{\gamma}_{q\ell}(\cI)
  =\cl
  \left\{
    \bigl(q_{\gamma,K}(\Q),\ell_{\gamma,K}(\Q)\bigr):
    \Q\in\cM(\cI)
  \right\},
  \label{eq:general-joint-set}
\end{equation}
where \(\cl\) denotes Euclidean closure.  The closure matters because the event indicator jumps at \(K\): a sequence of admissible measures can slide mass toward the threshold and change the event probability only in the limit, so the raw image need not be closed.  Conversely, any point outside \(\Theta^{\gamma}_{q\ell}(\cI)\) is at a positive distance from every pair an admissible measure can generate.

Conditional depth is a ratio and needs its own set,
\begin{equation}
  \Theta^{\gamma}_{s}(\cI)
  =
  \cl\left\{
    K-\frac{A_{\gamma+1}(\Q)}{A_\gamma(\Q)}:
    \Q\in\cM(\cI),\ A_\gamma(\Q)>0
  \right\},
  \label{eq:general-depth-set}
\end{equation}
which is empty if no admissible measure gives the event positive weighted mass.  Dividing the shortfall interval by the probability interval can overstate the range, because the four endpoints need not be driven by the same measure.  The set in \eqref{eq:general-depth-set} keeps numerator and denominator tied to a common measure.

\begin{theorem}[Identified sets under a power kernel]
\label{thm:sharp-set}
Fix \(\gamma>0\) and \(K\in(0,1)\).  Suppose \(B<\infty\) and \(\cM(\cI)\) is nonempty.  Then \(\Theta^{\gamma}_{q\ell}(\cI)\) is a nonempty compact convex subset of
\[
\{(q,\ell):0\leq q\leq1,\ 0\leq\ell\leq Kq\}.
\]
If some admissible \(\Q\) has \(A_\gamma(\Q)>0\), then \(\Theta^{\gamma}_{s}(\cI)\) is a nonempty compact interval contained in \([0,K]\).
\end{theorem}

Convexity is not automatic, because mixing two measures does not mix their probability--shortfall pairs linearly: each pair carries its own normalizing denominator \(M_\gamma(\Q)\).  The proof gets around this by mixing the measures with weights inversely proportional to their normalizing moments, which delivers any prescribed convex combination of the two pairs; admissibility survives because the quote and forward restrictions are linear in the measure.  Appendix~\ref{app:proofs} records the algebra together with the compactness argument.

The boundary of the joint set records where the maintained restrictions bind, and we call it the cross section's physical crash frontier, the physical counterpart of the risk-neutral crash frontier of \citet{zhuang2026likelydeepsharpjoint}.  By Theorem~\ref{thm:sharp-set} the set is compact and convex, so it is recovered from its support function \citep{rockafellar1970convex}.  For a compact convex \(\Theta\subset\R^2\) and a direction \(u=(u_q,u_\ell)\in\R^2\), write
\begin{equation}
  H_\Theta(u)
  =
  \sup_{(q,\ell)\in\Theta}
  \left(u_q q+u_\ell\ell\right).
  \label{eq:general-support-function}
\end{equation}
The value \(H_\Theta(u)\) locates the supporting line with outward normal \(u\), and the set is the intersection of the half planes those lines bound:
\[
\Theta = \bigcap_{\lVert u\rVert_2=1} \left\{(q,\ell): u_q q+u_\ell\ell\leq H_\Theta(u)\right\}.
\]
The directions \((\pm1,0)\) and \((0,\pm1)\) return the marginal interval endpoints, and intermediate directions carry the joint information that the marginals miss.

Adding quotes can only shrink the sets.  If \(\cI_1\subseteq\cI_2\), then \(\cM(\cI_2)\subseteq\cM(\cI_1)\), so both identified sets under \(\cI_2\) are subsets of their counterparts under \(\cI_1\).  This monotonicity organizes the empirical design in Section~\ref{sec:empirical}.  The sets are monotone in the support bound as well: enlarging \(B\) enlarges the admissible set, so both identified sets expand weakly with \(B\), and Section~\ref{sec:tail-results} traces this monotone family.

\subsection{Exact computation under \(\gamma=2\)}
\label{sec:finite-representation}

Theorem~\ref{thm:sharp-set} holds for every \(\gamma>0\), and integer values of \(\gamma\) admit exact finite representations.  The empirical benchmark is \(\gamma=2\), adopted from the one-year market-return calibration in \citet{martinshi2026forecasting}.  It is an external benchmark for the 23--37 day horizons studied here, rather than a parameter estimated on these cross sections.  Appendix~\ref{app:gamma-one} develops the \(\gamma=1\) log investor.

From here on we fix \(\gamma=2\), drop the superscript on the sets, and write
\[
q_{2,K},\quad \ell_{2,K},\quad s_{2,K},\quad M(\Q)=\E_\Q[X^2],\quad \Theta_{q\ell}(\cI),\quad \Theta_{s}(\cI).
\]
The exact statements below maintain one regularity condition.

\begin{assumption}[Interior regularity (a Slater-type condition)]
\label{ass:regularity}
Some \(\Q^\circ\in\cM(\cI)\) prices every retained quote strictly inside its interval, and there are bounded open intervals \(U_-\subset(0,1)\) and \(U_+\subset(1,B)\), disjoint from all strike and threshold boundaries, with \(\Q^\circ(U_-)>0\) and \(\Q^\circ(U_+)>0\).
\end{assumption}

The condition gives room to perturb an admissible measure without leaving \(\cM(\cI)\): the strict quote slack absorbs small changes in option values, and the interior mass on both sides of the forward restores the mean.  Relative to the strict quote interior condition of \citet{zhuang2026likelydeepsharpjoint}, the new ingredient is the interior mass, which the mean restoration needs.  The limiting arguments in Appendix~\ref{app:proofs} use exactly this room, and Section~\ref{sec:reported-statistics} verifies the condition cross section by cross section.

Computing a support value means maximizing a linear functional of \((q_{2,K},\ell_{2,K})\) over all admissible measures, an infinite-dimensional problem whose objectives are ratios of moments.  Two transformations turn it into a finite convex program: a change of scale removes the ratio, and a partition of the support at the strikes compresses the measure into finitely many moments whose feasible values are characterized exactly.

The change of scale is the linear-fractional normalization of \citet{charnescooper1962}.  Define \(\rho=M(\Q)^{-1}\) and the scaled measure \(d\widetilde{\Q}=\rho\,d\Q\).  The mass, forward, and second-moment restrictions become
\begin{equation}
  \int d\widetilde{\Q}=\rho,\qquad
  \int X\,d\widetilde{\Q}=\rho,\qquad
  \int X^2\,d\widetilde{\Q}=1,
  \label{eq:cc-normalization}
\end{equation}
and the two crash coordinates become linear in the scaled measure:
\begin{equation}
  q_{2,K}
  =
  \int_{[0,K]}X^2\,d\widetilde{\Q},
  \qquad
  \ell_{2,K}
  =
  \int_{[0,K]}(K-X)X^2\,d\widetilde{\Q}.
  \label{eq:scaled-targets}
\end{equation}
Each quote band is multiplied by \(\rho\), which keeps the quote restrictions linear in the pair \((\widetilde{\Q},\rho)\).  Nothing is lost in the change of variables: since \(0\leq X\leq B\) and \(\E_\Q[X]=1\), the global moment satisfies \(1\leq M(\Q)\leq B\), so \(\rho\in[1/B,1]\) and \(\Q\) is recovered from \((\widetilde{\Q},\rho)\) by division.

The second transformation partitions the support at the strikes.  Let \(0=b_0<b_1<\cdots<b_J=B\) be the ordered distinct elements of \(\{0,K,B\}\cup\{k_i:i\in\cI,\ 0<k_i<B\}\), and call the intervals \(I_j=[b_j,b_{j+1}]\) cells.  Let \(j_K\) be the index with \(b_{j_K}=K\), so that \(I_{j_K}\) is the cell beginning at the threshold, and let \(\mathcal J_-=\{j:j<j_K\}\) collect the cells below it.  Decompose the scaled measure into local components,
\[
\widetilde{\Q}=\sum_{j}\mu_j,
\]
where \(\mu_j\) is a nonnegative measure supported on \(I_j\); such a decomposition always exists.  An atom at a shared strike endpoint may be assigned to either adjacent cell, and no statement below depends on that choice, because every retained payoff is continuous there.  The threshold is the one exception: any atom at \(K\) is assigned to the cell beginning at \(K\) and tracked by the explicit coordinate introduced below.  For each cell \(j\) and integer order \(r\), define the scaled cell moment
\[
y_{rj}=\int x^r\,d\mu_j(x),
\]
and let the vector \(\boldsymbol y\) collect these moments, with orders \(r=0,1,2,3\) on the cells in \(\mathcal J_-\) and orders \(r=0,1,2\) elsewhere.  These are all the moments the problem uses.  The normalization \eqref{eq:cc-normalization} reads
\begin{equation}
  \sum_j y_{0j}=\rho,\qquad
  \sum_j y_{1j}=\rho,\qquad
  \sum_j y_{2j}=1.
  \label{eq:global-moment-constraints}
\end{equation}
Every retained call payoff is affine on every cell, because its only kink sits at a strike and every strike is a cell boundary.  The scaled value of quote \(i\) is therefore linear in the cell masses and first moments,
\begin{equation}
  C_i(\boldsymbol y)
  =
  \sum_{j:b_j\geq k_i}(y_{1j}-k_i y_{0j}),
  \qquad
  \rho\,\underline c_i
  \leq C_i(\boldsymbol y)
  \leq\rho\,\overline c_i,
  \label{eq:scaled-quote-constraints}
\end{equation}
and the targets \eqref{eq:scaled-targets} are linear in the second and third moments of the cells below the threshold.

The remaining question is which vectors \(\boldsymbol y\) come from actual local measures.  This is the truncated Hausdorff moment problem, and on bounded cells it has a complete answer.  On a cell \([a,b]\) with orders up to three, a vector \((y_0,y_1,y_2,y_3)\) is the moment sequence of some nonnegative measure on \([a,b]\) if and only if two localizing matrices are positive semidefinite:
\begin{equation}
  \begin{pmatrix}
    y_1-a y_0 & y_2-a y_1\\
    y_2-a y_1 & y_3-a y_2
  \end{pmatrix}
  \succeq0,
  \qquad
  \begin{pmatrix}
    b y_0-y_1 & b y_1-y_2\\
    b y_1-y_2 & b y_2-y_3
  \end{pmatrix}
  \succeq0.
  \label{eq:order-three-cone}
\end{equation}
On a cell with orders up to two, the conditions are the moment matrix
\begin{equation}
  \begin{pmatrix}
    y_0 & y_1\\
    y_1 & y_2
  \end{pmatrix}
  \succeq0,
  \label{eq:order-two-cone}
\end{equation}
combined with the support inequalities
\begin{equation}
  ay_0\leq y_1\leq by_0,
  \qquad
  y_2\leq(a+b)y_1-ab y_0.
  \label{eq:order-two-support}
\end{equation}
These conditions are necessary and sufficient on a bounded interval \citep[see][]{kreinnudelman1977,bertsimaspopescu2002}.  A two-by-two positive semidefinite constraint is in turn one second-order cone constraint \citep{alizadehgoldfarb2003}, since
\[
  \begin{pmatrix}u&v\\v&w\end{pmatrix}\succeq0
  \quad\Longleftrightarrow\quad
  u\geq0,\quad w\geq0,\quad
  \lVert(2v,u-w)\rVert_2\leq u+w.
\]
Cell by cell, the feasible moment vectors form an intersection of linear restrictions and second-order cone constraints.

For explicit event bookkeeping, we retain the threshold-atom coordinate of Section~\ref{sec:physical-transform}.  Let \(\widetilde a_K\geq0\) denote the scaled event mass at \(X=K\).  The atom is assigned to the threshold cell \(I_{j_K}\), whose moment conditions then apply net of the atom; Appendix~\ref{app:proofs} records this step.  The targets are the linear functions
\begin{equation}
  q_{2,K}
  =\sum_{j\in\mathcal J_-}y_{2j}+K^2\widetilde a_K,
  \qquad
  \ell_{2,K}
  =\sum_{j\in\mathcal J_-}(K y_{2j}-y_{3j}),
  \label{eq:cell-targets}
\end{equation}
so the atom adds probability and no shortfall, exactly as in \eqref{eq:general-threshold-atom}.

The program is now complete.  Its decision variables are the cell moments \(\boldsymbol y\) together with the two scalars \(\rho\) and \(\widetilde a_K\).  Its constraints are of three kinds: the global equalities \eqref{eq:global-moment-constraints}; the pair of scaled quote inequalities \eqref{eq:scaled-quote-constraints} for each retained quote; and one moment cone per cell, \eqref{eq:order-three-cone} on the cells in \(\mathcal J_-\) and \eqref{eq:order-two-cone}--\eqref{eq:order-two-support} elsewhere, with the threshold cell taken net of the atom.  Maximizing the linear objective \(u_q q_{2,K}+u_\ell\ell_{2,K}\) from \eqref{eq:cell-targets} is therefore a finite second-order cone program (SOCP) for each direction \(u\).\footnote{For integer \(\gamma\), the same construction applies with cell moments up to order \(\gamma+1\) below the threshold and order \(\gamma\) elsewhere.  The truncated moment conditions become positive semidefinite constraints on larger moment and localizing matrices \citep{kreinnudelman1977,bertsimaspopescu2002}, so the support values are computed by semidefinite programming.  For \(\gamma=1,2\), the largest matrices in this construction are two by two, so second-order cones suffice.}

\begin{corollary}[Exact computation under \(\gamma=2\)]
\label{cor:gamma-two-socp}
Suppose Assumption~\ref{ass:regularity} holds.  Then \(\Theta_{q\ell}(\cI)\) equals the image of the SOCP feasible region in \((\boldsymbol y,\rho,\widetilde a_K)\) under the linear map \eqref{eq:cell-targets}, and every support value \(H_{\Theta_{q\ell}}(u)\) is both the SOCP optimum and the supremum over admissible measures.  If some admissible measure has \(A_2(\Q)>0\), analogous programs give the endpoints of the closed interval \(\Theta_s(\cI)\).
\end{corollary}

Appendix~\ref{app:proofs} proves the corollary in two directions: every admissible measure produces a feasible vector, and every feasible vector is produced by an admissible measure or by a limit of admissible measures, with Assumption~\ref{ass:regularity} supplying the approximating sequences.  Conditional depth needs a second normalization, because its denominator is the truncated rather than the global second moment; scaling by \(A_2(\Q)^{-1}\) makes the depth objective linear, and the same cells and moment cones apply.  Unlike \(\rho\), the depth scaling need not be bounded above, so the depth programs need not attain their optimal values, which nevertheless equal the endpoints of the closed depth interval.

\begin{corollary}[Scenario feasibility]
\label{cor:scenario-audit}
Under Assumption~\ref{ass:regularity}, a pair \((q^\circ,\ell^\circ)\) lies in \(\Theta_{q\ell}(\cI)\) if and only if the SOCP feasible region is nonempty once the two linear targets \eqref{eq:cell-targets} are fixed at \((q^\circ,\ell^\circ)\).  Screening a candidate crash scenario is one conic feasibility problem.
\end{corollary}

The corollary follows from Corollary~\ref{cor:gamma-two-socp}: the identified set equals the image of the feasible region under a linear map, so membership of a candidate pair is equivalent to feasibility with the image coordinates pinned.

Assembling the identified sets from the programs is then direct.  Each direction \(u\) is one SOCP whose solution returns the support value \(H_{\Theta_{q\ell}}(u)\) together with an optimal pair \((q,\ell)\) on the frontier.  The four coordinate directions give the marginal probability and shortfall intervals, and the two depth programs give the depth interval.  In exact arithmetic, a finite grid of unit directions produces an inner and an outer polygon that sandwich the set.  The convex hull of the optimal pairs is contained in \(\Theta_{q\ell}(\cI)\), while the intersection over the grid of the supporting half planes \(\{(q,\ell):u_q q+u_\ell\ell\leq H_{\Theta_{q\ell}}(u)\}\) contains it.  The gap between the polygons measures what the grid can still miss, so the grid is doubled until the gap falls below a preset tolerance, and the figures plot the outer polygon.  Half planes computed for a coarser quote set remain valid under every finer set, so intersecting the accumulated half planes yields frontiers that are nested by construction.  For numerical solves, these containment statements require feasible primal points and valid dual upper bounds; the geometric gap alone measures only the reconstruction from the computed supports.

\subsection{What the right tail identifies}
\label{sec:tail-theory}

The support cap matters because it controls the denominator of the \(\gamma=2\) transform.  Quotes can be dense near the crash threshold and still leave the far right tail weakly constrained, and the finite cap is what prevents a small probability at a very large value of \(X\) from carrying an arbitrarily large second moment.  Let \(\cM_\infty(\cI)\) denote the admissible measures of \eqref{eq:admissible-measures} with \([0,B]\) replaced by \([0,\infty)\) and with finite second moment.

\begin{proposition}[Zero lower bound under a perturbable right tail]
\label{prop:unbounded-tail}
Suppose \(\cM_\infty(\cI)\) contains a measure \(\Q^\circ\) whose option values lie strictly inside every retained quote interval and which has positive mass on bounded open intervals \(U_-\subset(0,1)\) and \(U_+\subset(1,\infty)\).  Then
\begin{equation}
  \sup_{\Q\in\cM_\infty(\cI)}M(\Q)=\infty,
  \qquad
  \inf_{\Q\in\cM_\infty(\cI)}q_{2,K}(\Q)
  =\inf_{\Q\in\cM_\infty(\cI)}\ell_{2,K}(\Q)=0,
  \label{eq:zero-lower-bound}
\end{equation}
and the closure of the joint image \(\{(q_{2,K}(\Q),\ell_{2,K}(\Q)):\Q\in\cM_\infty(\cI)\}\) contains the origin.
\end{proposition}

The mechanism is the dilution channel of the introduction.  Mass of order \(H^{-3/2}\) placed at a state \(H\) far in the right tail changes fixed-strike option values by a vanishing amount but raises the second moment by order \(H^{1/2}\).  A small transfer between \(U_-\) and \(U_+\) restores total mass and the forward restriction, and strict quote slack absorbs the remaining changes in option values.  Since \(A_2(\Q)\leq K^2\), the inflated denominator drives \(q_{2,K}\) to zero and, along the same sequence, the unconditional shortfall as well.  The global factor \(M(\Q)\) cancels from conditional depth, so denominator inflation alone does not determine its bounds; its set can still respond to right-tail restrictions through the quote and forward constraints.

\begin{remark}[Log utility is the dividing line]
\label{rem:log-dividing-line}
The collapse is not specific to \(\gamma=2\).  For \(\gamma\leq1\), Jensen's inequality and the forward restriction give \(M_\gamma(\Q)\leq1\) for every admissible measure, so mass far in the right tail cannot inflate the denominator of the transform.  For every \(\gamma>1\), suppose the strict-interior and interior-mass hypotheses of Proposition~\ref{prop:unbounded-tail} hold in the class of unit-mean, quote-compatible measures with finite \(\gamma\)th moment.  Its construction then applies with mass \(H^{-(1+\gamma)/2}\) at the state \(H\).  Option values then change by \(O(H^{(1-\gamma)/2})\), the moment \(M_\gamma\) grows without bound, and \(q_{\gamma,K}\) and \(\ell_{\gamma,K}\) vanish along the sequence.  The dilution channel opens exactly when the investor is more risk averse than the log investor.
\end{remark}

\section{Data and Design}
\label{sec:empirical}

\subsection{SPX option cross sections}
\label{sec:data}

The data are end-of-day bid and ask quotes on S\&P 500 index options from OptionMetrics, covering January 2013 through August 2023.  We retain quotes with positive strikes, nonnegative bids, positive asks, uncrossed bid--ask intervals, and positive trading volume.  Quotes enter the analysis only through these intervals.  The sample uses exact calendar horizons of 23, 30, and 37 days and keeps at most one cross section per horizon in a calendar week, observed on Wednesdays.  It contains 1,115 horizon-specific observations covering 537 distinct calendar weeks, and each cross section carries its own discount factor and forward.

\subsection{Crash event and nested quote sets}
\label{sec:information-sets}

The main event is a terminal decline of at least 10 percent from each cross section's forward, so \(K=0.90\).  Appendix~\ref{app:supplementary} examines the more extreme event \(K=0.85\).  The main support condition \(B=5\) allows the terminal index level to reach five times its forward over a 23--37 day horizon.  It is a loose economic bound that controls the kernel normalization and restricts the admissible measures.  Section~\ref{sec:tail-results} examines tighter and looser caps and the unbounded specification.

The design expands the quote set in four steps.  The first set, \(\cI^{\mathrm{pair}}\), contains the closest retained put strike on each side of \(K\).  The second, \(\cI^{\mathrm{local}}\), adds the six nearest additional put strikes, for eight in total.  The third, \(\cI^{\mathrm{wing}}\), contains the complete retained out-of-the-money (OTM) put wing.  The final set, \(\cI^{\mathrm{OTM}}\), which we call the full OTM set, adds all retained OTM calls and so comprises every retained OTM quote.  Thus
\begin{equation}
  \cI^{\mathrm{pair}}
  \subseteq
  \cI^{\mathrm{local}}
  \subseteq
  \cI^{\mathrm{wing}}
  \subseteq
  \cI^{\mathrm{OTM}}.
  \label{eq:nested-information}
\end{equation}
By the monotonicity of Section~\ref{sec:identified-set}, every interval and every support value contracts weakly along this sequence, so the design separates information near the event threshold from information in the rest of the put wing and on the call side.

\subsection{Reported statistics, verification, and inference}
\label{sec:reported-statistics}

For each cross section we solve the programs of Section~\ref{sec:finite-representation}.  Their sharp interpretation uses Assumption~\ref{ass:regularity}; solver acceptance and regularity are separate requirements.  All results condition on \(\gamma=2\), \(K=0.90\), \(B=5\), and each cross section's fixed discounting and forward.  Summaries are formed separately by horizon: the reported lower endpoint is the median of the lower endpoints across cross sections, and similarly for the upper endpoint.  The main measure of information is the proportional reduction in the width of the crash probability interval.  For nested quote sets \(\cI_1\subseteq\cI_2\) with widths \(w_1>0\) and \(w_2\), the reduction within a cross section is \(\Delta=1-w_2/w_1\), and we report its median across cross sections within each horizon.  Sampling variation of these medians is assessed with a circular block bootstrap on calendar weeks, using eight-week blocks and 2,000 replications stratified by horizon \citep{politisromano1992,sharipovwendler2013}.  Bracketed entries are percentile 95 percent intervals.

Assumption~\ref{ass:regularity} is verifiable one cross section at a time.  For each cross section we attempt to construct an admissible measure that prices every retained quote strictly inside its interval and places mass on both sides of the forward away from the strike and threshold boundaries.  The construction succeeds in 1,111 of the 1,115 cross sections, and Table~\ref{tab:regularity} in Appendix~\ref{app:supplementary} reports the counts by horizon.

The full OTM programs pass the numerical acceptance checks in 1,098 cross sections, which form the samples of Table~\ref{tab:headline}.  Comparisons across quote sets use the 976 cross sections passing those checks under all four quote sets.  Cross sections that fail verification for a program are dropped from the affected statistics rather than approximated, and each table and figure reports its own counts.

\section{What Do the Traded Quotes Reveal?}
\label{sec:marginal-results}

\subsection{Main identified sets}
\label{sec:headline-results}

Table~\ref{tab:headline} reports the main results.  Under all retained OTM quotes, the median identified crash probability interval rises with the horizon, reaching 1.97 to 4.53 percent at 37 days.  Even with every retained OTM quote in use, the upper endpoint remains two to three times the lower endpoint, so the quotes still admit economically different assessments of the same event.

Conditional depth is identified less tightly, with median intervals spanning roughly 2.4 to 7.4 percent of additional decline below the 10 percent threshold.  The contrast between the two coordinates is not specific to this event: at \(K=0.85\) the probability intervals are narrower and the depth intervals wider, with the same ordering of quote-set contributions (Appendix~\ref{app:supplementary}).

\begin{table}[t]
\caption{Physical crash risk and the information in SPX option quotes}
\label{tab:headline}
\centering
\small
\begin{threeparttable}
\textit{Panel A: Main identified sets}\\[2pt]
\begin{tabular}{lrrrrr}
\toprule
Horizon & Cross sections & OTM quotes & $q_{2,K}$ & $\ell_{2,K}$ & $s_{2,K}$ \\
\midrule
23 days & 402 & 105.0 & [0.71, 2.11] & [0.047, 0.061] & [2.44, 7.35] \\
30 days & 372 & 144.0 & [1.42, 3.29] & [0.086, 0.105] & [2.81, 7.00] \\
37 days & 324 & 95.5 & [1.97, 4.53] & [0.127, 0.152] & [2.94, 7.44] \\
\bottomrule
\end{tabular}
\vspace{6pt}

\textit{Panel B: Median widths and reductions}\\[2pt]
\begin{tabular}{lrrrrr}
\toprule
Horizon & Put pair & Local eight & Put wing & All OTM &
Total reduction [95\% CI] \\
\midrule
23 days & 7.05 & 1.86 & 1.39 & 1.34 & 79.9 [78.4, 81.7] \\
30 days & 12.14 & 2.87 & 1.91 & 1.82 & 84.4 [82.9, 85.3] \\
37 days & 13.80 & 3.42 & 2.50 & 2.36 & 81.8 [80.3, 83.0] \\
\bottomrule
\end{tabular}
\vspace{6pt}

\textit{Panel C: Marginal source of probability identification}\\[2pt]
\begin{tabular}{lrrr}
\toprule
Added quote layer & 23 days & 30 days & 37 days \\
\midrule
Pair to local eight
& 69.6 [66.1, 72.0]
& 72.2 [69.1, 74.8]
& 72.5 [70.2, 75.0] \\
Local eight to put wing
& 28.0 [25.9, 31.0]
& 35.2 [32.7, 36.9]
& 26.9 [25.3, 28.4] \\
Put wing to all OTM
& 2.6 [2.2, 3.3]
& 4.2 [3.3, 5.2]
& 5.0 [4.1, 6.4] \\
\bottomrule
\end{tabular}
\begin{tablenotes}[flushleft]
\footnotesize
\item Notes: One S\&P 500 cross section per week and horizon, January 2013 through August 2023.  All results use \(\gamma=2\), \(K=0.90\), \(B=5\), and each cross section's fixed discounting and forward.  Panel A reports the number of cross sections, median retained OTM quote counts, and medians of interval endpoints across cross sections, in percent.  The \(s_{2,K}\) medians use 402, 371, and 324 cross sections.  Panel B reports median interval widths for the crash probability, in percentage points, and the median proportional reduction from the put pair to all OTM quotes, in percent.  Panel C reports median proportional reductions relative to the preceding quote set, in percent.  Brackets are 95 percent calendar-week block bootstrap intervals.
\end{tablenotes}
\end{threeparttable}
\end{table}

Panel B of Table~\ref{tab:headline} shows how much ambiguity remains under each quote set.  All retained OTM quotes cut the put pair widths by a median of 80 to 84 percent across horizons, and the bootstrap intervals are narrow relative to these magnitudes.

\subsection{Physical crash risk over time}
\label{sec:time-results}

Figure~\ref{fig:time-series} plots the probability intervals under the full OTM set by horizon.  They rise in every major stress episode of the sample, and the width often rises with the level.  Both endpoints move sharply upward around the COVID-19 episode at every horizon, so stress raises the implied crash risk and, at the same time, the range of crash probabilities the quotes admit.

\begin{figure}[t]
\centering
\includegraphics[width=\textwidth]{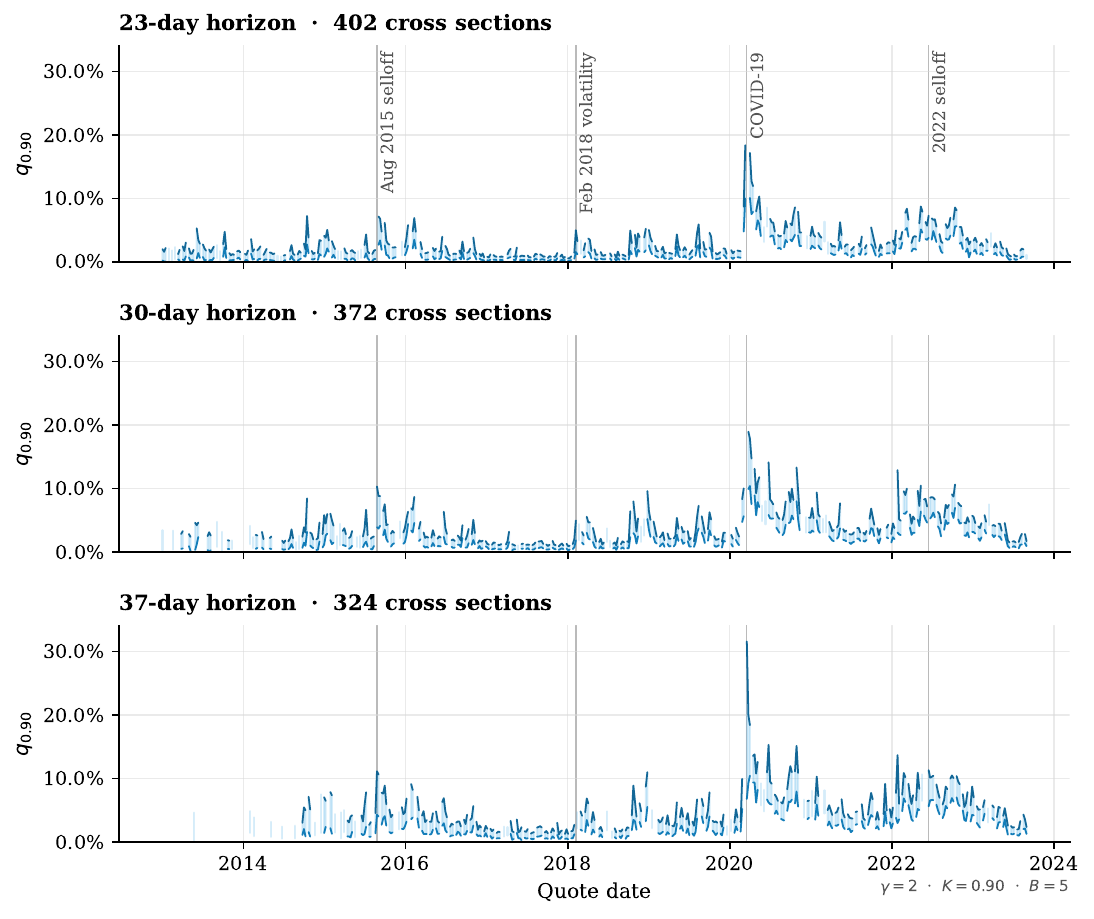}
\caption{Time variation in quote-compatible physical crash probabilities.  Shaded bands are weekly identified intervals for the physical probability of \(X\leq0.90\), in percent, under the benchmark kernel and all retained OTM quotes.  The cross sections cover exact horizons of 23, 30, and 37 days with 402, 372, and 324 cross sections.  Missing weeks appear as gaps.  Results use \(\gamma=2\) and \(B=5\).}
\label{fig:time-series}
\end{figure}

\subsection{Where the identifying information comes from}
\label{sec:information-results}

Figure~\ref{fig:identification-gains} separates the level of the remaining ambiguity from the contribution of each quote layer.  The largest change occurs when the six additional nearby puts join the put pair, the rest of the put wing remains informative, and OTM calls add little beyond the complete put wing (Table~\ref{tab:headline}, Panel C).  The identification of a 10 percent decline therefore comes mainly from put prices, with the strongest contribution near the threshold.

These are marginal contributions along the nested sequence, and a reverse exercise confirms the same concentration.  Deleting the far put wing from the full quote set widens the median interval, deleting the OTM calls barely moves it, and deleting the six nearby puts leaves the median width unchanged, because neighboring wing strikes substitute for them.  The quoted intervals themselves anchor the identification: doubling every bid--ask spread moves the endpoints an order of magnitude more than retaining only every other strike.  Appendix~\ref{app:supplementary} reports the magnitudes.

\begin{figure}[t]
\centering
\includegraphics[width=\textwidth]{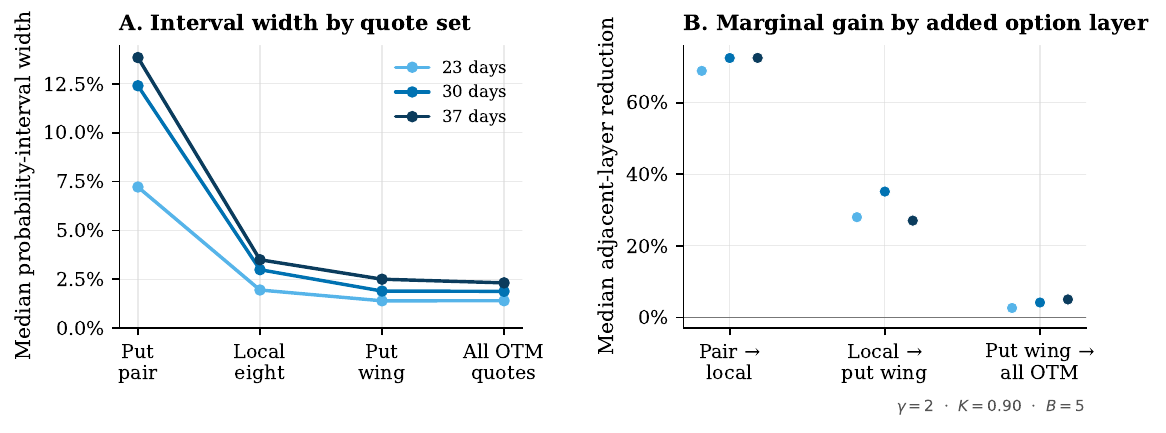}
\caption{Identification gains across nested quote sets.  Panel A reports the median width of the identified crash probability interval under each quote set and horizon, in percentage points.  Panel B reports the median proportional reduction relative to the preceding quote set, in percent.  Results use \(\gamma=2\), \(K=0.90\), and \(B=5\).}
\label{fig:identification-gains}
\end{figure}

\subsection{The unobserved right tail}
\label{sec:tail-results}

The main specification places \(X\) in \([0,5]\).  Figure~\ref{fig:tail-audit} compares that choice with tighter and looser support conditions on a fixed robustness subsample of 180 cross sections, 60 per horizon.  The subsample is spaced evenly over the sample period and includes cross sections adjacent to the August 2015, February 2018, March 2020, and June 2022 stress dates.  For finite caps between two and ten, the results are nearly unchanged.

Under the hypotheses of Proposition~\ref{prop:unbounded-tail}, removing the cap without further tail control makes the lower bounds on both probability and unconditional shortfall zero.  The unbounded moment program in Appendix~\ref{app:unbounded-outer} is used as an outer relaxation; its computed probability lower endpoint is zero in every subsample cross section.  For depth, the true unbounded interval contains the interval under $B=5$ and lies within the relaxed interval.  Proposition~\ref{prop:unbounded-outer} therefore bounds any increase in the true depth width by the increase in the relaxed width.  Figure~\ref{fig:tail-audit} reports computed width ratios close to one.

\begin{figure}[t]
\centering
\includegraphics[width=\textwidth]{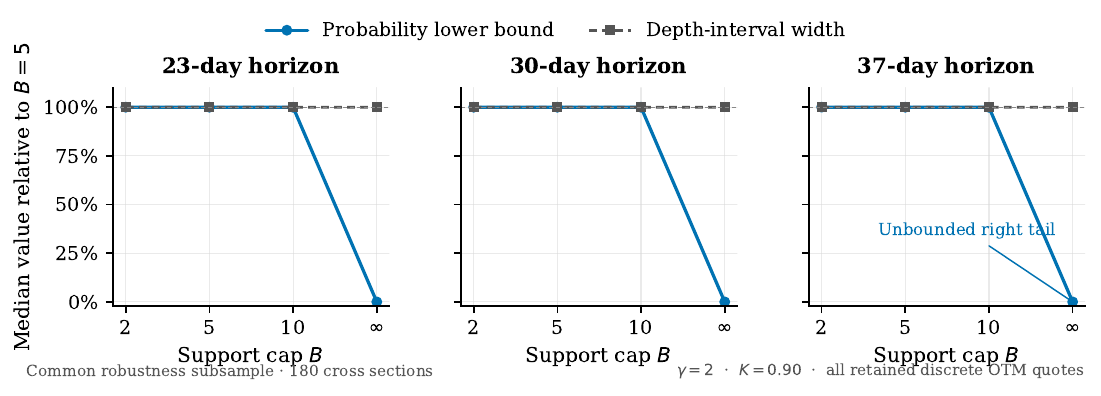}
\caption{Tail dependence of crash probability and conditional depth.  The robustness subsample contains 180 cross sections, 60 at each horizon.  The blue series reports the median ratio \(q_{2,K}^-(B)/q_{2,K}^-(5)\), where \(q_{2,K}^-(B)\) is the probability lower endpoint under support cap \(B\).  The gray dashed series reports the median ratio of conditional depth interval widths.  Finite caps use all retained OTM quotes.  Moving from \(B=2\) to \(B=10\) lowers a cross section's probability lower bound by a median of 0.0008 percentage points and by 0.0053 at the 95th percentile.  The unbounded point uses the moment relaxation in Appendix~\ref{app:unbounded-outer}: its computed probability lower endpoint is zero, and the gray series uses the outer depth interval.  The largest computed outer depth width ratio relative to \(B=5\) is 1.0002.  The exact-arithmetic inclusion is given in Proposition~\ref{prop:unbounded-outer}; the plotted ratios use numerical optima.  Results use \(\gamma=2\) and \(K=0.90\).}
\label{fig:tail-audit}
\end{figure}

\FloatBarrier

\section{Jointly Attainable Crash Scenarios}
\label{sec:joint-results}

Each of the four marginal endpoints is attained or approached by admissible measures, and the measures or sequences need not coincide, so a probability and shortfall pair drawn from the marginal rectangle may describe no single admissible distribution.  This section examines the physical crash frontiers, first for selected cross sections and then for the full sample.  Each plotted frontier follows the construction of Section~\ref{sec:finite-representation}, with the direction grid refined until the computed inner--outer gap is at most 1 percent of the cross section's put pair scale, the marginal interval widths under the put pair (Appendix~\ref{app:supplementary}).

Figure~\ref{fig:market-states} covers four prespecified 30-day market states, each matched to the closest admissible cross section: the 2015 selloff, a calm market in late 2019, the COVID-19 crash, and the 2022 selloff.  The identified sets contract in the same order in all four panels.  Two puts leave a large range of probability--shortfall combinations, nearby puts remove most of that range, and the complete put wing and the OTM calls leave a much smaller set.  The location of the surviving set tracks market conditions: the calm-market set lies close to the origin, and the 2015 and COVID-19 sets allow larger probability and shortfall values.

\begin{figure}[t]
\centering
\includegraphics[width=\textwidth]{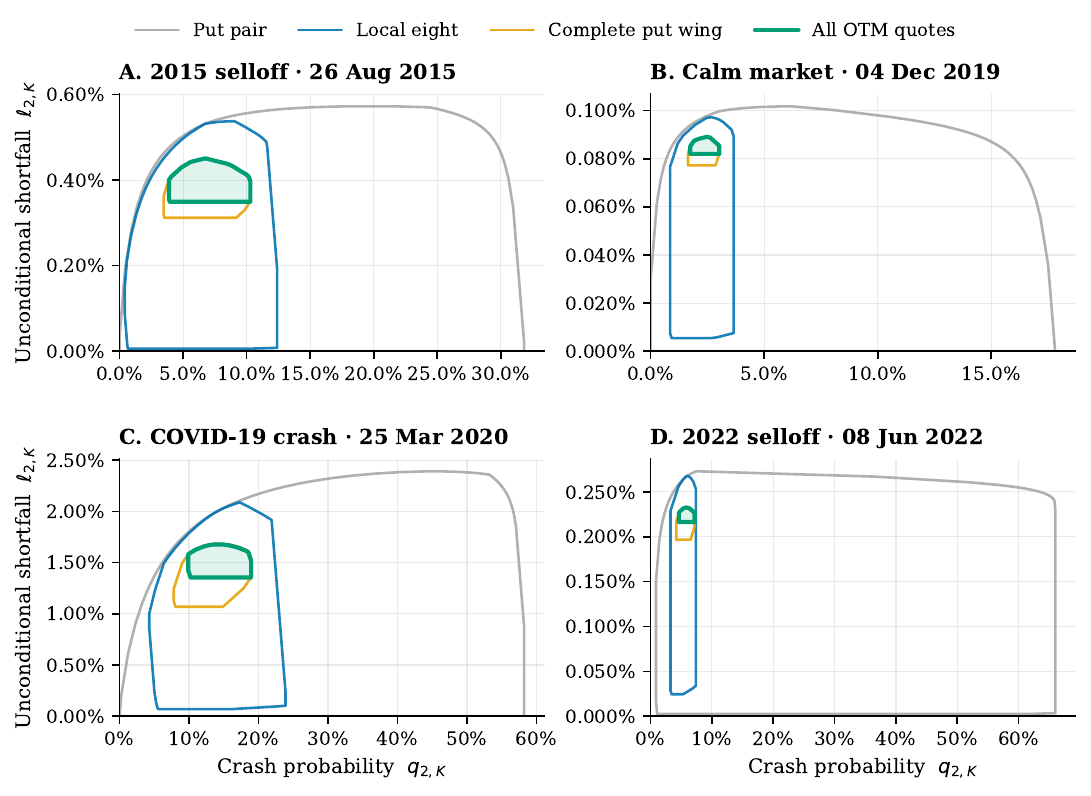}
\caption{Joint probability--shortfall sets across market states.  Each panel overlays the envelopes of supporting half planes of \(\Theta_{q\ell}(\cI)\) for the put pair, local eight, complete put wing, and all OTM quotes.  Each successive envelope intersects all half planes accumulated through that layer, so the sets are nested by construction.  Probability and unconditional shortfall are in percent, at a 30-day horizon with \(\gamma=2\), \(K=0.90\), and \(B=5\).  The cross section dates are 26 August 2015, 4 December 2019, 25 March 2020, and 8 June 2022.  The COVID-19 and 2022 targets were 18 March 2020 and 15 June 2022.  Each panel uses its own axes.}
\label{fig:market-states}
\end{figure}

Figure~\ref{fig:covid-horizons} repeats the exercise for three horizon-specific observations around the COVID-19 crash.  The dates differ across panels, so their differences combine market conditions and horizon effects.  The 37-day set reaches the largest shortfall values among the three and has a broad probability range under sparse quotes.  Additional puts contract the feasible region at every horizon, and the sets that remain once every OTM quote is included still differ in level and shape.

\begin{figure}[t]
\centering
\includegraphics[width=\textwidth]{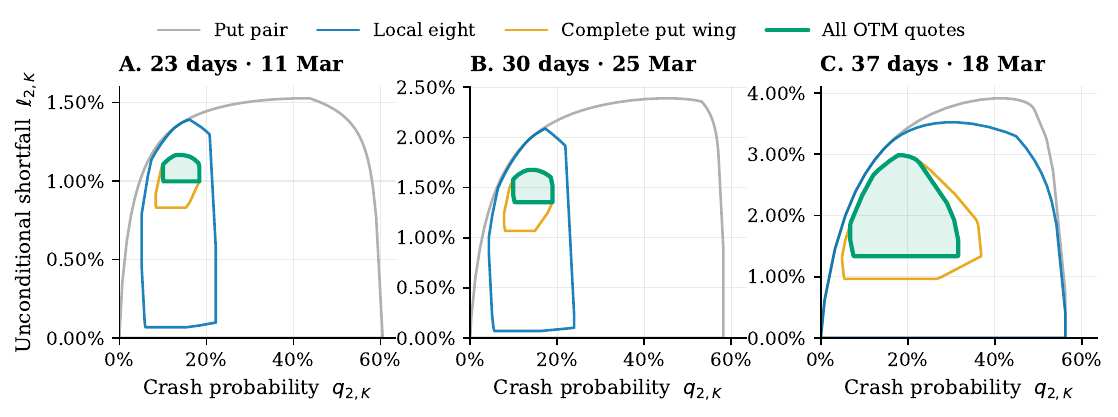}
\caption{Joint probability--shortfall sets near the COVID-19 crash.  Each panel overlays the envelopes of supporting half planes for the put pair, local eight, complete put wing, and all OTM quotes at one horizon.  Probability and unconditional shortfall are in percent.  The panels use the closest admissible 23-, 30-, and 37-day observations around the 18 March 2020 target, dated 11 March, 25 March, and 18 March 2020.  Results use \(\gamma=2\), \(K=0.90\), and \(B=5\).  Observation dates and vertical scales differ across panels.}
\label{fig:covid-horizons}
\end{figure}

Across the 976 cross sections shared by the four quote sets, the joint restriction is strongest where shortfall is stressed.  In the two directions that place positive weight on shortfall, the horizon-level median gaps between the full OTM joint set and its marginal rectangle run from 2.7 to 4.0 percent of the put pair scale (Appendix~\ref{app:joint-directions}).  The same directional statistics show where each quote layer acts: nearby puts deliver most of the contraction in high-probability directions, while the rest of the put wing dominates in low-probability directions.

\FloatBarrier

\section{Conclusion}
\label{sec:conclusion}

Finite option quotes are informative about short-horizon physical crash risk, and the physical crash frontier measures how informative.  Mapping every quote-compatible distribution through the power utility investor's pricing kernel separates what traded prices imply from what interpolation and tail assumptions add.  In the S\&P 500 cross sections studied here, most of the width of the put pair probability interval disappears once all retained OTM quotes are used, the largest contribution comes from puts near the threshold, and the sets rise and widen together in market stress.

Under the strict-interior and interior-mass hypotheses of Proposition~\ref{prop:unbounded-tail}, finite quotes without additional tail control cannot supply a positive floor on physical crash probability or unconditional shortfall when $\gamma>1$.  The support cap used here prevents the global normalizing moment from diverging; a finite upper bound on that moment would also block this dilution channel.  The denominator cancels from conditional depth.  Across the support conditions examined, the computed depth intervals change little, including the outer interval from the unbounded relaxation.  A positive probability floor in this setting therefore depends on tail information as well as traded prices.

The frontier turns this evidence into a discipline for scenario analysis.  It retains only the probability and shortfall combinations that a single admissible distribution can generate or approach, and its exclusions are strongest where scenarios stress shortfall.  Under the maintained kernel, event, and support conditions, any estimate that fills in the curve between and beyond the quoted strikes while pricing the same quotes within their spreads implies physical crash risk inside the frontier.

\clearpage
\bibliographystyle{plainnat}
\bibliography{references}

\appendix

\section{Proofs and Complete Programs}
\label{app:proofs}

\subsection{Positive-return and threshold closure}
\label{app:positive-closure}

\begin{lemma}[Approximation away from zero]
\label{lem:positive-return-closure}
Under Assumption~\ref{ass:regularity}, every \(\Q\in\cM(\cI)\) is the weak limit of quote-compatible measures supported on \((0,B]\).  For every fixed \(r>0\), their global and truncated \(r\)th moments converge to those of \(\Q\).
\end{lemma}

\begin{proof}
Form the mixture \((1-\delta_\varepsilon)\Q+\delta_\varepsilon\Q^\circ\), where \(\delta_\varepsilon\to0\) and \(\varepsilon/\delta_\varepsilon\to0\), and move all mass at zero in the mixture to \(\varepsilon>0\).  This preserves total mass and changes the mean and each retained option value by \(O(\varepsilon)\).  A transfer of order \(O(\varepsilon)\) between \(U_+\) and \(U_-\) restores the mean; the mass available in these intervals is of order \(\delta_\varepsilon\), so the transfer preserves nonnegativity for small \(\varepsilon\).  Strict quote slack inherited from the mixture absorbs the remaining payoff changes.

The resulting measures have strictly positive support.  For each \(r>0\), mixing and the mean-restoring transfer change both global and truncated moments by \(O(\delta_\varepsilon+\varepsilon)\), since their integrands are bounded on \([0,B]\).  Taking \(\varepsilon<K\), the move from zero to \(\varepsilon\) adds at most \(\varepsilon^r\) to either moment.  Both moments therefore converge, despite the jump of the truncated integrand at \(K\).
\end{proof}

\subsection{Proof of Theorem~\ref{thm:sharp-set}}

\begin{proof}[Convexity and compactness of the joint set]
For an admissible measure \(\Q_j\), let \(d_j=M_\gamma(\Q_j)>0\), and set
\[
  \boldsymbol n_j
  =
  \begin{pmatrix}
    A_{\gamma}(\Q_j)\\
    K A_{\gamma}(\Q_j)-A_{\gamma+1}(\Q_j)
  \end{pmatrix},
\]
so that the probability--shortfall image of \(\Q_j\) is \(\boldsymbol w_j=\boldsymbol n_j/d_j\).  Fix \(\theta\in[0,1]\) and define
\[
  \alpha
  =
  \frac{\theta d_1}
       {\theta d_1+(1-\theta)d_0}.
\]
The mixture \(\Q_\alpha=\alpha\Q_0+(1-\alpha)\Q_1\) is admissible because the quote and moment restrictions are linear in the measure, and
\[
  \frac{\alpha\boldsymbol n_0+(1-\alpha)\boldsymbol n_1}
       {\alpha d_0+(1-\alpha)d_1}
  =
  \theta\boldsymbol w_0+(1-\theta)\boldsymbol w_1.
\]
The attainable image is therefore convex, and its closure is convex and nonempty.  Since \(\ell_{\gamma,K}=q_{\gamma,K}s_{\gamma,K}\) with \(s_{\gamma,K}\in[0,K]\) whenever \(q_{\gamma,K}>0\), the closure is contained in \(\{(q,\ell):0\leq q\leq1,\ 0\leq\ell\leq Kq\}\).  It is closed and bounded, hence compact.
\end{proof}

\begin{proof}[The conditional depth interval]
Let \(\cM_+=\{\Q\in\cM(\cI):A_\gamma(\Q)>0\}\), a convex set.  For \(\Q_0,\Q_1\in\cM_+\) and \(\alpha\in[0,1]\), the mixture \(\Q_\alpha=\alpha\Q_0+(1-\alpha)\Q_1\) satisfies
\[
  s_{\gamma,K}(\Q_\alpha)
  =
  \frac{
    \alpha A_{\gamma}(\Q_0)s_{\gamma,K}(\Q_0)
    +(1-\alpha)A_{\gamma}(\Q_1)s_{\gamma,K}(\Q_1)
  }{
    \alpha A_{\gamma}(\Q_0)
    +(1-\alpha)A_{\gamma}(\Q_1)
  }.
\]
As \(\alpha\) varies, this expression spans the interval between \(s_{\gamma,K}(\Q_0)\) and \(s_{\gamma,K}(\Q_1)\).  The image is therefore an interval contained in \([0,K]\), and its closure is a nonempty compact interval, which is \(\Theta^{\gamma}_{s}(\cI)\).
\end{proof}

\subsection{Threshold bookkeeping and the depth programs}
\label{app:complete-programs}

Recall that \(j_K\) indexes the cell beginning at \(K\).  In the program of Section~\ref{sec:finite-representation}, the moment conditions on this cell apply to the residual moments
\[
  \overline y_{0\,j_K}=y_{0\,j_K}-\widetilde a_K,\qquad
  \overline y_{1\,j_K}=y_{1\,j_K}-K\widetilde a_K,\qquad
  \overline y_{2\,j_K}=y_{2\,j_K}-K^2\widetilde a_K,
\]
which satisfy \eqref{eq:order-two-cone}--\eqref{eq:order-two-support} with \(a=K\) and \(b=b_{j_K+1}\), together with \(\widetilde a_K\geq0\).  The atom is part of \(y_{r\,j_K}\), so it enters the global constraints \eqref{eq:global-moment-constraints} and the quote values \eqref{eq:scaled-quote-constraints} without further adjustment, while the targets \eqref{eq:cell-targets} count it as event mass.  An empty sum in \eqref{eq:scaled-quote-constraints} is interpreted as zero when \(k_i\geq B\).  The unbounded relaxation and the scope of its bounds are stated in Appendix~\ref{app:unbounded-outer}.

Conditional depth uses a second normalization.  When \(A_2(\Q)>0\), define \(\lambda=A_2(\Q)^{-1}\) and \(d\widehat{\Q}=\lambda\,d\Q\).  The mass, forward, and truncated second-moment restrictions become
\begin{equation}
  \int d\widehat{\Q}=\lambda,\qquad
  \int X\,d\widehat{\Q}=\lambda,\qquad
  \int_{[0,K]}X^2\,d\widehat{\Q}=1,
  \label{eq:depth-normalization}
\end{equation}
each quote band is multiplied by \(\lambda\), and the objective is linear:
\begin{equation}
  s_{2,K}=\int_{[0,K]}(K-X)X^2\,d\widehat{\Q}.
  \label{eq:depth-scaled}
\end{equation}
With cell moments \(\widehat y_{rj}=\int x^r\,d\widehat\mu_j\) taken from a local decomposition of \(\widehat{\Q}\), and with \(\widehat a_K\geq0\) the scaled event atom, the finite constraints are
\begin{equation}
  \sum_j\widehat y_{0j}=\lambda,\qquad
  \sum_j\widehat y_{1j}=\lambda,\qquad
  \sum_{j\in\mathcal J_-}\widehat y_{2j}
  +K^2\widehat a_K=1,
  \label{eq:depth-global-constraints}
\end{equation}
\begin{equation}
  \widehat C_i(\widehat{\boldsymbol y})
  =
  \sum_{j:b_j\geq k_i}
  (\widehat y_{1j}-k_i\widehat y_{0j}),
  \qquad
  \lambda\underline c_i
  \leq\widehat C_i(\widehat{\boldsymbol y})
  \leq\lambda\overline c_i,
  \label{eq:depth-quote-constraints}
\end{equation}
with the same cellwise moment cones and threshold treatment as before, and the linear objective
\begin{equation}
  s_{2,K}
  =
  \sum_{j\in\mathcal J_-}
  (K\widehat y_{2j}-\widehat y_{3j}).
  \label{eq:depth-cell-objective}
\end{equation}
No global second-moment normalization is imposed.

\paragraph{Eliminating the explicit atom.}
The atom coordinate is redundant in the closed moment representation.  Set $j=j_K$, add $K^r\widetilde a_K$ to $y_{r,j-1}$ for $r=0,1,2,3$, subtract $K^r\widetilde a_K$ from $y_{rj}$ for $r=0,1,2$, and set $\widetilde a_K$ to zero.  The left cell gains an endpoint atom and the right cell retains its residual moments, so both moment cones remain satisfied.  Global moments and quote values are unchanged, as are the targets: the added second moment replaces $K^2\widetilde a_K$ in $q$, while its shortfall contribution is $K^3\widetilde a_K-K^3\widetilde a_K=0$.  The same transformation applies to the depth normalization using $\widehat a_K$.  We keep these coordinates to record the event convention explicitly.

\subsection{Proof of Corollary~\ref{cor:gamma-two-socp}}
\label{app:exactness}

\begin{proof}
The feasible region is compact.  It is closed, being cut out by finitely many linear equalities and inequalities and by second-order cone constraints.  For boundedness, fix a feasible point.  The cellwise moment conditions are sufficient, so the cell moments, net of the designated atom on the threshold cell, are moments of nonnegative measures on the closed cells.  Summing the cells and adding the atom yields a nonnegative measure on \([0,B]\) with total mass \(\rho\), first moment \(\rho\), and second moment one.  The Cauchy--Schwarz inequality gives \(\rho^2\leq\rho\cdot1\), so \(\rho\leq1\), and the bounded support gives \(0\leq y_{rj}\leq B^r y_{0j}\leq B^r\) for every stored order together with \(0\leq\widetilde a_K\leq y_{0\,j_K}\leq1\).  Every coordinate of a feasible point is therefore bounded.

The conic image of the feasible region under the linear targets \eqref{eq:cell-targets} is therefore compact.  We show that every attainable pair lies in it, so that, being closed, it contains \(\Theta_{q\ell}(\cI)\), and that every point of it is a limit of attainable pairs, so that it is contained in \(\Theta_{q\ell}(\cI)\).

First take any admissible \(\Q\) and set \(\rho=M(\Q)^{-1}\in[1/B,1]\) and \(d\widetilde{\Q}=\rho\,d\Q\).  Decompose \(\widetilde{\Q}\) into local components, assigning the atom at \(K\), if any, to the cell beginning at \(K\), and let \(\widetilde a_K\) be that scaled atom.  The resulting \((\boldsymbol y,\rho,\widetilde a_K)\) satisfies every constraint: the global equalities are \eqref{eq:cc-normalization}, the quote inequalities are the scaled quote bands, the cells below \(K\) carry genuine moment sequences of measures on closed intervals, and the threshold cell satisfies its conditions net of the designated atom.  The targets \eqref{eq:cell-targets} equal \((q_{2,K}(\Q),\ell_{2,K}(\Q))\) by \eqref{eq:general-threshold-atom}.

Conversely, take any feasible \((\boldsymbol y,\rho,\widetilde a_K)\).  The truncated Hausdorff conditions \eqref{eq:order-three-cone}--\eqref{eq:order-two-support} supply, cell by cell, a nonnegative measure on the closed cell with the prescribed moments.  Adding the designated atom \(\widetilde a_K\) at \(K\) and summing the cells produces a nonnegative measure \(\widetilde{\Q}'\) that satisfies the mass, forward, second-moment, and quote restrictions, since each restriction is an integral of a continuous function against the cell moments.  Where a reconstructed cell measure places mass at a shared endpoint \(b_j\neq K\), the location is unambiguous relative to the threshold, so neither the constraints nor the targets depend on which cell carries it.  The support inequalities give \(1=\sum_j y_{2j}\leq B\sum_j y_{1j}=B\rho\), so \(\rho\geq1/B>0\), and the measure \(d\Q'=d\widetilde{\Q}'/\rho\) has mass one, mean one, second moment \(1/\rho\), and satisfies every quote interval.

One discrepancy can remain.  The reconstruction of the cell beginning at \(K\) may itself place mass at \(K\) beyond the designated atom, whereas the targets \eqref{eq:cell-targets} treat that residual mass as lying above the threshold.  Move the residual mass to \(K+\varepsilon\) and repair the measure with the mixing construction of Lemma~\ref{lem:positive-return-closure}: blend in weight \(\delta_\varepsilon\) of the interior measure \(\Q^\circ\), transfer order \(\varepsilon\) between \(U_-\) and \(U_+\) to restore the mean, and use the inherited strict quote slack to keep every quote value inside its interval for small \(\varepsilon\).  The resulting measures \(\Q_\varepsilon\) are admissible, and their probability and shortfall converge to the targets of the feasible point.  Every feasible point is therefore attained or approached, so the conic image equals \(\Theta_{q\ell}(\cI)\), and the support values follow by linear optimization over it.

For the depth programs, let \(S\) denote the set of attainable conditional depths and \(V\) the set of objective values of feasible points.  Any admissible \(\Q\) with \(A_2(\Q)>0\) produces a feasible pair \((\widehat{\boldsymbol y},\lambda)\) with \(\lambda=A_2(\Q)^{-1}\) whose objective value is \(s_{2,K}(\Q)\), so \(S\subseteq V\).  Conversely, every feasible point carries a finite scalar \(\lambda\), which the unit constraint in \eqref{eq:depth-global-constraints} forces to be positive.  If its threshold-cell reconstruction carries no residual mass at \(K\), dividing \(\widehat{\Q}'\) by \(\lambda\) yields an admissible measure with \(A_2=1/\lambda\) whose conditional depth equals the objective \eqref{eq:depth-cell-objective}; in general, the same \(K+\varepsilon\) construction yields admissible measures whose depths converge to the objective value.  Every value in \(V\) is therefore a limit of values in \(S\), and
\[
  S\subseteq V\subseteq\cl S=\Theta_s(\cI).
\]
Suprema and infima agree across this chain, so the optimal values of the two conic programs equal the endpoints of the compact interval \(\Theta_s(\cI)\).  These optimal values need not be attained by feasible points: when \(A_2\downarrow0\) along a sequence of admissible measures, \(\lambda\) diverges, and no feasible point need exist whose objective value equals the endpoint.
\end{proof}

\subsection{Unbounded moment relaxation}
\label{app:unbounded-outer}

For the unbounded specification, partition at $0$, $K$, and all retained positive strikes.  The last cell is $[a,\infty)$.  On that cell retain
\[
\begin{pmatrix}y_0&y_1\\y_1&y_2\end{pmatrix}\succeq0,
\qquad y_1\geq ay_0,
\]
applied to the residual moments if the cell begins at $K$, and omit the upper support inequalities.  All bounded cells retain their original cones.  Keep the global and quote constraints, with $\rho\geq0$ for the joint program and $\lambda\geq0$ for the depth programs; the cap-dependent bound $\rho\geq1/B$ is removed.

These are necessary moment conditions, used here as an outer relaxation.  No representation by an actual half-line measure is asserted for every feasible vector.  In particular, the joint relaxation admits vectors with $\rho=0$, zero mass and first moment, and positive second moment on the last cell.  Such a point cannot be inverted to a probability measure by division by $\rho$.

\begin{proposition}[Outer bounds for unbounded conditional depth]
\label{prop:unbounded-outer}
Fix the quotes, $K$, and $\gamma=2$.  Let $\Theta_s(B)$ be the depth interval under a finite support cap $B$, and let $\Theta_s(\infty)$ be the closure of depths generated by $\cM_\infty(\cI)$ with $A_2>0$.  Let $\widehat\Theta_s(\infty)$ be the closed interval whose endpoints are the infimum and supremum of the unbounded depth relaxation.  If $\Theta_s(B)$ is nonempty, then
\[
\Theta_s(B)\subseteq\Theta_s(\infty)
\subseteq\widehat\Theta_s(\infty)\subseteq[0,K].
\]
Consequently, if $w_B=\operatorname{width}(\Theta_s(B))>0$,
\[
1\leq\frac{\operatorname{width}(\Theta_s(\infty))}{w_B}
\leq\frac{\operatorname{width}(\widehat\Theta_s(\infty))}{w_B}.
\]
\end{proposition}

\begin{proof}
Any admissible measure supported on $[0,B]$ is also admissible for the unbounded problem and has finite second moment, which gives the first inclusion after taking closures.  For any measure in $\cM_\infty(\cI)$ with $A_2>0$, scaling by $A_2^{-1}$ yields finite cell moments satisfying every constraint of the unbounded depth relaxation.  Its depth is therefore between the infimum and supremum of that program, and the second inclusion follows by closure.  On every bounded cell below $K$, the moment cone represents a nonnegative measure and $0\leq(K-x)x^2\leq Kx^2$.  The event second-moment normalization, including the nonnegative designated atom, is one.  Hence every feasible objective lies in $[0,K]$.  The width inequalities follow from interval inclusion.
\end{proof}

The same measure-to-moment map gives outer support bounds for the unbounded joint image.  Its probability and shortfall lower bounds are zero under Proposition~\ref{prop:unbounded-tail}, independently of whether the full relaxation is sharp.  When the finite-cap regularity condition holds, Corollary~\ref{cor:gamma-two-socp} identifies the finite-cap depth endpoints exactly.  Thus the width comparison above can be used without proving exactness of the unbounded moment relaxation.  Numerical certification of a width ratio also requires error bounds for both endpoint programs.

\subsection{Proof of Proposition~\ref{prop:unbounded-tail}}

\begin{proof}
Let \(\nu_-\) and \(\nu_+\) be the conditional distributions of \(\Q^\circ\) on \(U_-\) and \(U_+\), with means \(\mu_-<1<\mu_+\).  For \(H\) above the retained strikes and \(U_+\), set
\[
  \varepsilon_H=H^{-3/2},\qquad
  a_H=\varepsilon_H\frac{H-\mu_+}{\mu_+-\mu_-},\qquad
  b_H=\varepsilon_H\frac{H-\mu_-}{\mu_+-\mu_-}.
\]
Let \(\delta_H\) denote unit mass at \(H\), and define
\[
  \Q_H
  =
  \Q^\circ+\varepsilon_H\delta_H+a_H\nu_- - b_H\nu_+.
\]
The coefficients satisfy \(\varepsilon_H+a_H-b_H=0\) and \(\varepsilon_HH+a_H\mu_- - b_H\mu_+=0\), so \(\Q_H\) has mass one and mean one.  Since \(b_H\to0\) and \(\Q^\circ(U_+)>0\), the measure is nonnegative for all sufficiently large \(H\).

For a retained call payoff \(h_i(x)=(x-k_i)^+\),
\[
  \E_{\Q_H}[h_i(X)]-\E_{\Q^\circ}[h_i(X)]
  =
  \varepsilon_H h_i(H)
  +a_H\E_{\nu_-}[h_i(X)]
  -b_H\E_{\nu_+}[h_i(X)]
  =O(H^{-1/2}).
\]
There are finitely many quotes, so strict interior slack implies that \(\Q_H\) satisfies every quote interval for all sufficiently large \(H\).  The intervals \(U_-\) and \(U_+\) are bounded, so
\[
  M(\Q_H)
  =
  M(\Q^\circ)+\varepsilon_HH^2+O(H^{-1/2})
  =
  M(\Q^\circ)+H^{1/2}+O(H^{-1/2})
  \longrightarrow\infty.
\]
Finally, \(A_2(\Q_H)\leq K^2\), and since \(K-X\leq K\) on the event, also \(KA_2(\Q_H)-A_3(\Q_H)\leq KA_2(\Q_H)\leq K^3\).  Hence
\[
  0\leq q_{2,K}(\Q_H)
  \leq\frac{K^2}{M(\Q_H)}
  \longrightarrow0,
  \qquad
  0\leq \ell_{2,K}(\Q_H)
  \leq\frac{K^3}{M(\Q_H)}
  \longrightarrow0,
\]
so the probability--shortfall pair of \(\Q_H\) converges to the origin, which therefore lies in the closure of the joint image.
\end{proof}

\section{The Log-Investor Benchmark}
\label{app:gamma-one}

\subsection{Mathematical specialization}

When \(\gamma=1\), the forward restriction fixes the denominator of the transform \eqref{eq:physical-transform} at \(M_1(\Q)=\E_\Q[X]=1\), so the dilution channel of Proposition~\ref{prop:unbounded-tail} is absent and the tail acts only through the quote and forward restrictions (Remark~\ref{rem:log-dividing-line}).  The probability and unconditional shortfall are therefore linear in the unscaled measure,
\begin{equation}
  q_{1,K}(\Q)=A_{1}(\Q),\qquad
  \ell_{1,K}(\Q)=K A_{1}(\Q)-A_{2}(\Q),
  \label{eq:gamma-one-linear-targets}
\end{equation}
and, when \(A_{1}(\Q)>0\), conditional depth is \(s_{1,K}(\Q)=K-A_{2}(\Q)/A_{1}(\Q)\).  No change of scale is needed for the joint set.  Use the same cells as Section~\ref{sec:finite-representation} and a local decomposition \(\Q=\sum_j\mu_j\), with strike-endpoint atoms assigned to either adjacent cell and any atom at \(K\) assigned to the cell beginning at \(K\).  For each \(\mu_j\), define the cell mass \(z_j\) and first moment \(m_j\), together with the second moment \(v_j\) on the cells in \(\mathcal J_-\).  The global restrictions are
\begin{equation}
  \sum_j z_j=1,\qquad
  \sum_j m_j=1.
  \label{eq:gamma-one-global}
\end{equation}
Cells below the threshold carry the complete order-two conditions \eqref{eq:order-two-cone}--\eqref{eq:order-two-support} in \((z_j,m_j,v_j)\), and cells at or above \(K\) need only \(az_j\leq m_j\leq bz_j\).  With \(j_K\) again indexing the cell beginning at \(K\) and \(a_K\) the threshold mass assigned to the closed event,
\begin{equation}
  0\leq a_K\leq z_{j_K},\qquad
  m_{j_K}\leq b_{j_K+1}z_{j_K}-(b_{j_K+1}-K)a_K,
  \label{eq:gamma-one-threshold-atom}
\end{equation}
each quote value is linear as in \eqref{eq:scaled-quote-constraints} with \((z_j,m_j)\) in place of \((y_{0j},y_{1j})\) and unscaled bands, and the targets are
\begin{equation}
  q_{1,K}
  =\sum_{j\in\mathcal J_-}m_j+Ka_K,\qquad
  \ell_{1,K}
  =\sum_{j\in\mathcal J_-}(Km_j-v_j).
  \label{eq:gamma-one-cell-targets}
\end{equation}
Under Assumption~\ref{ass:regularity}, the closure argument of Appendix~\ref{app:exactness} identifies its image under \eqref{eq:gamma-one-cell-targets} with the closed log-investor joint set.  Conditional depth again uses a second normalization: with \(\lambda=A_{1}(\Q)^{-1}\) and \(d\widehat\Q=\lambda\,d\Q\), the mass and forward integrals both equal \(\lambda\), the unit constraint is \(\int_{[0,K]}X\,d\widehat\Q=1\), and the objective \(\int_{[0,K]}(K-X)X\,d\widehat\Q\) is linear.  The scaled quote bands, cell moment conditions, and threshold-atom treatment carry over from Appendix~\ref{app:complete-programs}, with first and second moments in the roles that second and third moments play there.

\subsection{Comparison with the \(\gamma=2\) kernel}
\label{app:kernel-comparison}

\begin{proposition}[Kernel comparison]
\label{prop:kernel-comparison}
Fix \(K\in(0,1)\) and consider a common quote-compatible domain of unit-mean measures with finite second moments.  For every \(\Q\) in this domain with positive truncated moments,
\begin{equation}
  q_{2,K}(\Q)\leq Kq_{1,K}(\Q)\leq q_{1,K}(\Q),
  \qquad
  \ell_{2,K}(\Q)\leq K\ell_{1,K}(\Q),
  \qquad
  s_{2,K}(\Q)\leq s_{1,K}(\Q).
  \label{eq:gamma-comparison}
\end{equation}
\end{proposition}

\begin{proof}
On the event \(X\leq K\), \(X^2\leq KX\).  Also \(M(\Q)\geq M_1(\Q)^2=1\).  Hence
\[
  q_{2,K}
  =\frac{A_{2}}{M(\Q)}
  \leq A_{2}
  \leq K A_{1}
  =Kq_{1,K}
  \leq q_{1,K}.
\]
The same argument applied to \((K-X)X^2\leq K(K-X)X\) gives \(\ell_{2,K}\leq K\ell_{1,K}\).  The truncated moment sequence is log-convex by the Cauchy--Schwarz inequality, so \(A_{2}^2\leq A_{1}A_{3}\).  Dividing by \(A_{1}A_{2}>0\) and subtracting from \(K\) gives \(s_{2,K}\leq s_{1,K}\).
\end{proof}

Taking extrema in \eqref{eq:gamma-comparison} compares corresponding endpoints: each \(\gamma=2\) probability and shortfall endpoint lies below \(K\) times its \(\gamma=1\) counterpart, and each depth endpoint lies below its counterpart on the common positive-probability domain.  The inequalities impose no ordering across different endpoints, and the optimizing measures can differ.

\subsection{Full-sample comparison}
\label{app:kernel-empirics}

Table~\ref{tab:kernel-comparison} reports both kernels on the same cross sections, the same event \(K=0.90\), the same support cap \(B=5\), and the same full OTM quote sets.  The sample keeps the cross sections where both kernels' programs solve, a weaker criterion than the numerical acceptance checks behind Table~\ref{tab:headline}, so the counts differ slightly.  This table is a supplementary sensitivity comparison; the sharp interpretation of each program still requires Assumption~\ref{ass:regularity}.  The \(\gamma=2\) intervals are narrower for every coordinate at every horizon.  Proposition~\ref{prop:kernel-comparison} orders endpoint levels, not widths, so the width ordering is an empirical regularity, and the two kernels agree on how identification varies across horizons and coordinates.

\begin{table}[t]
\caption{Kernel comparison on common cross sections}
\label{tab:kernel-comparison}
\centering
\small
\begin{threeparttable}
\begin{tabular}{lrrrrrrr}
\toprule
& & \multicolumn{2}{c}{$q$ width}
& \multicolumn{2}{c}{$\ell$ width}
& \multicolumn{2}{c}{$s$ width} \\
\cmidrule(lr){3-4}\cmidrule(lr){5-6}\cmidrule(lr){7-8}
Horizon & Cross sections & $\gamma=1$ & $\gamma=2$
& $\gamma=1$ & $\gamma=2$ & $\gamma=1$ & $\gamma=2$ \\
\midrule
23 days & 405 & 1.54 & 1.40 & 0.017 & 0.013 & 4.93 & 4.81 \\
30 days & 369 & 2.08 & 1.88 & 0.024 & 0.017 & 4.27 & 4.10 \\
37 days & 332 & 2.57 & 2.35 & 0.033 & 0.023 & 4.54 & 4.26 \\
\bottomrule
\end{tabular}
\begin{tablenotes}[flushleft]
\footnotesize
\item Notes: Median interval widths in percentage points, computed on the cross sections where both kernels solve, under all retained OTM quotes, \(K=0.90\), \(B=5\), and each cross section's fixed discounting and forward.
\end{tablenotes}
\end{threeparttable}
\end{table}

\section{Supplementary Evidence}
\label{app:supplementary}

\subsection{A more extreme threshold}
\label{app:extreme-threshold}

Table~\ref{tab:extreme-threshold} repeats the analysis under the full OTM set at \(K=0.85\).  The probability intervals are narrower than at \(K=0.90\) and the conditional depth intervals are wider: a more extreme event can have a narrow probability range and a wide range of depth conditional on occurring.  The quote layers contribute in the same order as at the main threshold.  Adding the six nearby puts cuts the median put pair widths by 70 to 75 percent, and the complete put wing and the calls deliver the remaining contraction to the widths in the table.

\begin{table}[t]
\caption{Identified sets for a 15 percent decline}
\label{tab:extreme-threshold}
\centering
\small
\begin{threeparttable}
\begin{tabular}{lrrrrr}
\toprule
Horizon & Cross sections & Median $q_{2,K}$ & $q$ width & Depth cross sections & $s$ width \\
\midrule
23 days & 404 & [0.14, 0.70] & 0.53 & 393 & 9.98 \\
30 days & 373 & [0.33, 1.08] & 0.76 & 362 & 7.70 \\
37 days & 330 & [0.47, 1.61] & 1.07 & 329 & 8.07 \\
\bottomrule
\end{tabular}
\begin{tablenotes}[flushleft]
\footnotesize
\item Notes: All retained OTM quotes, \(K=0.85\), \(B=5\).  Probability entries are medians of interval endpoints across cross sections, in percent.  Width entries are median interval widths, in percentage points.  Depth widths use the cross sections whose closed probability interval excludes zero.
\end{tablenotes}
\end{threeparttable}
\end{table}

\subsection{Quote perturbations and deletions}
\label{app:quote-perturbations}

Two exercises alter the information content of the quotes on a prespecified subset of 24 cross sections from the robustness subsample.  Doubling every bid--ask spread around its midpoint changes a probability endpoint by a median of 0.79 percentage points and a conditional depth endpoint by a median of 1.84 percentage points.  Retaining every other strike changes the corresponding endpoints by medians of 0.073 and 0.122 percentage points.

A reverse decomposition deletes groups of quotes from the full OTM set on the robustness subsample of 180 cross sections.  Removing the put wing beyond the local eight widens the median probability interval by 0.46, 0.76, and 0.52 percentage points at the three horizons.  Removing all OTM calls widens it by 0.05, 0.08, and 0.13 percentage points.  Removing the six nearby puts while keeping every other quote leaves the median width unchanged, because the neighboring wing strikes substitute for them.  Marginal contributions therefore depend on the base set: nearby puts carry most of the information relative to the sparse pair, and the wing carries information that survives even when the local strikes are dense.

\subsection{Verification of the regularity condition}
\label{app:regularity-table}

Table~\ref{tab:regularity} summarizes the verification of Assumption~\ref{ass:regularity} in each cross section, described in Section~\ref{sec:reported-statistics}.  The verification fails only in the 23- and 37-day cross sections of 20 January 2016 and the 30-day cross sections of 31 May 2017 and 7 September 2022.

\begin{table}[t]
\caption{Verification of the regularity condition by cross section}
\label{tab:regularity}
\centering
\small
\begin{threeparttable}
\begin{tabular}{lrrrr}
\toprule
Horizon & Cross sections & Verified & Rate (\%) & Median slack ($\times10^{-5}$) \\
\midrule
23 days & 406 & 405 & 99.8 & 1.1 \\
30 days & 376 & 374 & 99.5 & 1.1 \\
37 days & 333 & 332 & 99.7 & 1.5 \\
\bottomrule
\end{tabular}
\begin{tablenotes}[flushleft]
\footnotesize
\item Notes: A cross section is verified when an admissible measure prices every retained quote strictly inside its bid--ask interval and places mass on both sides of the forward away from strike and threshold boundaries.  Slack is the smallest distance from a constructed option value to its quote endpoints, in normalized price units.
\end{tablenotes}
\end{threeparttable}
\end{table}

\subsection{Directional joint restrictions}
\label{app:joint-directions}

Let \(q_{\mathrm{pair}}^-\) and \(\ell_{\mathrm{pair}}^-\) be the put pair lower endpoints of a cross section, and let \(w_q^{\mathrm{pair}}\) and \(w_\ell^{\mathrm{pair}}\) be the positive put pair widths.  Joint comparisons use the normalized coordinates
\begin{equation}
  \bar q
  =\frac{q-q_{\mathrm{pair}}^-}{w_q^{\mathrm{pair}}},
  \qquad
  \bar\ell
  =\frac{\ell-\ell_{\mathrm{pair}}^-}{w_\ell^{\mathrm{pair}}}.
  \label{eq:joint-normalization}
\end{equation}
Let \(\overline\Theta\) be the normalized joint set and \(\overline{\mathcal B}\) its marginal rectangle, with support functions \(H_{\overline\Theta}\) and \(H_{\overline{\mathcal B}}\) as in \eqref{eq:general-support-function}.  For the four diagonal directions
\[
  u\in
  \left\{
  \tfrac{(1,1)}{\sqrt2},
  \tfrac{(1,-1)}{\sqrt2},
  \tfrac{(-1,1)}{\sqrt2},
  \tfrac{(-1,-1)}{\sqrt2}
  \right\},
\]
the normalized support gap and the directional contraction between nested quote sets \(\cI_1\subseteq\cI_2\) are
\begin{equation}
  g(u)
  =
  100\left[
  H_{\overline{\mathcal B}}(u)
  -H_{\overline\Theta}(u)
  \right],
  \qquad
  c_{1\rightarrow2}(u)
  =
  100\left[
  H_{\overline\Theta(\cI_1)}(u)
  -H_{\overline\Theta(\cI_2)}(u)
  \right].
  \label{eq:directional-statistics}
\end{equation}

Figure~\ref{fig:directional-appendix} summarizes these statistics on the 976 cross sections shared by the four quote sets.  The gap between the joint set and its marginal rectangle concentrates in the two directions that place positive weight on shortfall, where the horizon-level median gaps run from 2.7 to 4.0 percent of the put pair scale.  Pooling all four directions, the median gap under the full OTM set is 0.52 percent.  The contraction statistics show where each layer acts.  Moving from the pair to the local eight produces the largest contractions in high-probability directions, with medians of 33 and 51 percent of the pair scale.  Moving from the local eight to the complete put wing dominates in low-probability directions, with medians of 8 and 43 percent.  The median contraction from the put pair to the full OTM set across the four directions is 50 percent of the pair scale.

\begin{figure}[t]
\centering
\includegraphics[width=\textwidth]{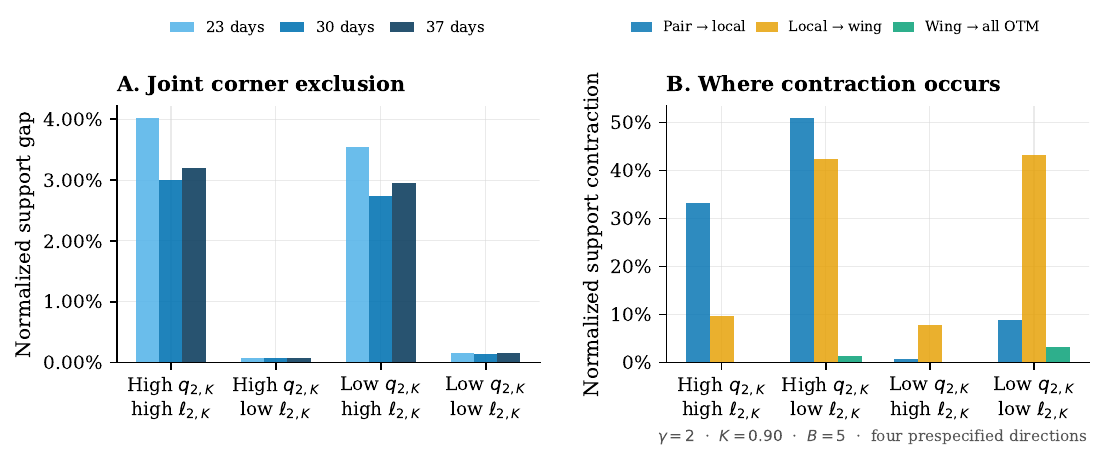}
\caption{Directional joint restrictions and quote-set contraction.  Panel A reports the median normalized support gap between the full OTM joint set and its marginal rectangle in four prespecified directions.  Panel B reports the median support contraction contributed by adjacent quote layers in the same directions.  Both panels use each cross section's put pair marginal widths for normalization and cover the 976 cross sections shared by the four quote sets.  Results use \(\gamma=2\), \(K=0.90\), and \(B=5\).}
\label{fig:directional-appendix}
\end{figure}

\subsection{Accuracy of the plotted joint sets}
\label{app:joint-geometry}

The joint sets of Section~\ref{sec:joint-results} are drawn as intersections of supporting half planes on an adaptive direction grid.  The grid doubles from 32 to at most 128 directions until the gap to the inner hull of support points falls below 1 percent of the put pair scale.  Table~\ref{tab:joint-geometry} reports the realized accuracy for each plotted cross section: the largest normalized gap across the four quote sets is at most 0.99 percent, which measures the reported geometric gap relative to the put pair widths.  It does not give a 1 percent relative error guarantee for the smaller full OTM set, and it excludes support-solver error.

\begin{table}[t]
\caption{Adaptive-grid accuracy of the plotted joint sets}
\label{tab:joint-geometry}
\centering
\small
\begin{threeparttable}
\begin{tabular}{llrrr}
\toprule
Panel & Date & Horizon & Directions & Max gap (\%) \\
\midrule
2015 selloff & 26 Aug 2015 & 30 days & 128 & 0.80 \\
Calm market & 4 Dec 2019 & 30 days & 128 & 0.93 \\
COVID-19 crash & 25 Mar 2020 & 30 days & 128 & 0.90 \\
2022 selloff & 8 Jun 2022 & 30 days & 128 & 0.98 \\
COVID-19, 23 days & 11 Mar 2020 & 23 days & 128 & 0.76 \\
COVID-19, 37 days & 18 Mar 2020 & 37 days & 64 & 0.99 \\
\bottomrule
\end{tabular}
\begin{tablenotes}[flushleft]
\footnotesize
\item Notes: For each plotted cross section, Directions is the largest adaptive grid level used across the four quote sets, and Max gap is the largest normalized support gap between the plotted outer polygon and the inner hull of its support points, in percent of the put pair scale.  The 30-day COVID-19 cross section appears in both Figures~\ref{fig:market-states} and~\ref{fig:covid-horizons}.
\end{tablenotes}
\end{threeparttable}
\end{table}

\end{document}